%% file: ROOT.tex
\documentclass[corr]{techreport}
\pdfoutput=1
\usepackage[pdftex]{graphicx,color}

\input{packages}

\input{defns}

\begin{document}

\title{Advanced Automata Minimization}


\authorinfo{Lorenzo Clemente}
           {LaBRI, University of Bordeaux I}
           {lorenzo.clemente@labri.fr}
\authorinfo{Richard Mayr}
           {University of Edinburgh}
           {http://homepages.inf.ed.ac.uk/rmayr}

\maketitle

\begin{abstract}
We present an efficient algorithm to reduce the size of nondeterministic 
B\"uchi word automata, while retaining their language.
Additionally, we describe methods to solve PSPACE-complete automata problems like
universality, equivalence and inclusion for much larger instances 
(1-3 orders of magnitude) than before. This can be used to
scale up applications of automata in formal verification tools
and decision procedures for logical theories.

The algorithm is based on new transition pruning techniques. 
These use criteria based on combinations of backward and forward
trace inclusions.
Since these relations are themselves PSPACE-complete, we describe
methods to compute good approximations of them in polynomial time.

Extensive experiments show that the average-case complexity of our algorithm
scales quadratically. The size reduction of the automata depends very much on
the class of instances, but our algorithm consistently outperforms all previous techniques
by a wide margin. We tested our algorithm on B\"uchi automata derived from
LTL-formulae, many classes of random automata and automata derived from
mutual exclusion protocols, and compared its performance to the well-known
automata tool GOAL \cite{GOAL_survey_paper}.
\end{abstract}




\input{introduction}

\input{preliminaries}

\input{minimization}

\input{lookahead}

\input{heavyandlight}

\input{inclusion}

\input{empirical}

\input{conclusion}

\bibliographystyle{abbrvnat}

\input{bibliography}
\newpage
\appendix
\input{appendix}

\end{document}

%% file: packages.tex
\usepackage{graphicx}
\usepackage{epsfig}
\usepackage{xcolor}
\usepackage{colortbl}
\usepackage{amsmath}
\usepackage{amssymb}
\usepackage{mathtools}
\usepackage{multirow}
\usepackage{makecell}
\usepackage[ruled,vlined,linesnumbered]{algorithm2e}
\usepackage{amsthm}
\usepackage{subfigure}

\usepackage{pgf}
\usepackage{tikz}

\usetikzlibrary{arrows,%
	shapes,%
	decorations,%
	decorations.pathmorphing,%
	decorations.text,%
	decorations.pathreplacing,%
	automata,%
	arrows,%
	chains,%
	matrix,%
	scopes,%
	shadows,%
	positioning,%
	patterns,%
	fit,%
	calc,%
	through,%
	fadings,%
	backgrounds}

\usepackage{stmaryrd}
\usepackage{paralist}
\usepackage{pslatex}
\usepackage{color}

\usepackage[textsize=small]{todonotes}

%% file: defns.tex
\newtheorem{theorem}{Theorem}[section]
\newtheorem{lemma}[theorem]{Lemma}

\newtheorem{corollary}[theorem]{Corollary}

\newenvironment{definition}[1][Definition]{\begin{trivlist}
\item[\hskip \labelsep {\bfseries #1}]}{\end{trivlist}}

\newenvironment{remark}[1][Remark]{\begin{trivlist}
\item[\hskip \labelsep {\bfseries #1}]}{\end{trivlist}}

\newcommand{\xparagraph}[1]{\paragraph{#1}}
\newcommand{\ignore}[1]{}

\newcommand{\goesto}[1]{\stackrel {#1} \longrightarrow}
\newcommand{\comesfrom}[1]{\stackrel {#1} \longleftarrow}

\newcommand{\symb}{\sigma}
\newcommand{\trans}[3]{{#1} \goesto {#2} {#3}} 
\newcommand{\prefix}[2]{#1[0..#2]}
\newcommand{\suffix}[2]{#1[#2..]}

\newcommand{\A}{\mathcal A}
\newcommand{\B}{\mathcal B}
\newcommand{\lang}[1]{\mathcal{L}(#1)}

\newcommand{\xsim}[1]{\sqsubseteq^{#1}}
\newcommand{\disim}{\xsim{\mathsf{di}}}
\newcommand{\desim}{\xsim{\mathsf{de}}}
\newcommand{\fsim}{\xsim{\mathsf f}}
\newcommand{\bwsim}{\xsim{\mathsf{bw}}}

\newcommand{\xincl}[1]{\subseteq^{#1}}
\newcommand{\xstrictincl}[1]{\subset^{#1}}
\newcommand{\directtraceinclusion}{\xincl{\mathsf{di}}}
\newcommand{\delayedtraceinclusion}{\xincl{\mathsf{de}}}
\newcommand{\fairtraceinclusion}{\xincl{\mathsf f}}
\newcommand{\strictfairtraceinclusion}{\xstrictincl{\mathsf f}}
\newcommand{\bwdirecttraceinclusion}{\xincl{\mathsf{bw}}}
\newcommand{\accblindbwdirecttraceinclusion}{\xincl{\mathsf{bw-}}}
\newcommand{\countingbwtraceinclusion}{\xincl{\mathsf{bw}\textrm{-}\mathsf{c}}}
\newcommand{\delayedfixedwordsimulation}{\xsim{\mathsf{fx\textrm-de}}}
\newcommand{\languageinclusion}{\xincl {}}
\newcommand{\languageequivalence}{\approx}
\newcommand{\strictdirecttraceinclusion}{\xstrictincl{\mathsf{di}}}
\newcommand{\strictbwdirecttraceinclusion}{\xstrictincl{\mathsf{bw}}}
\newcommand{\strictlanguageinclusion}{\strictfairtraceinclusion}

\newcommand{\brel}{R_b}
\newcommand{\frel}{R_{\!f}}
\newcommand{\prune}[2]{{\it Prune}({#1},{#2})}
\newcommand{\prunerel}{P}
\newcommand{\makeprunerel}[2]{\prunerel({#1},{#2})}
\newcommand{\xprune}[3]{{\it Prune}({#1},{#2},{#3})}

\newcommand{\id}{{\it id}} 

\newcommand{\strictxsim}[1]{\sqsubset^{#1}}
\newcommand{\strictdisim}{\strictxsim{\mathsf{di}}}
\newcommand{\strictdesim}{\strictxsim{\mathsf{de}}}
\newcommand{\strictfsim}{\strictxsim{\mathsf f}}
\newcommand{\strictbwsim}{\strictxsim{\mathsf{bw}}}

\newcommand{\ksim}[1]{\sqsubseteq^{#1}}
\newcommand{\kxsim}[2]{\ksim{#1\textrm - #2}}
\newcommand{\kbwsim}{\kxsim k {\mathsf{bw}}}

\newcommand{\transksimx}{\preceq^{k\textrm{-}x}}
\newcommand{\transkbwsim}{\preceq^{k\textrm{-}\mathsf{bw}}}
\newcommand{\accblindkbwsim}{\sqsubseteq^{k\textrm{-}\mathsf{bw-}}}
\newcommand{\accblindtranskbwsim}{\preceq^{k\textrm{-}\mathsf{bw-}}}
\newcommand{\countingkbwsim}{\sqsubseteq^{k\textrm{-}\mathsf{bw}\textrm{-}\mathsf{c}}}
\newcommand{\countingtranskbwsim}{\preceq^{k\textrm{-}\mathsf{bw}\textrm{-}\mathsf{c}}}
\newcommand{\transkdisim}{\preceq^{k\textrm{-}\mathsf{di}}}
\newcommand{\transkdesim}{\preceq^{k\textrm{-}\mathsf{de}}}
\newcommand{\transkfsim}{\preceq^{k\textrm{-}\mathsf f}}

\newcommand{\transkdisimnumber}[1]{\preceq^{{#1}\textrm{-}\mathsf{di}}}
\newcommand{\transkdesimnumber}[1]{\preceq^{{#1}\textrm{-}\mathsf{de}}}

\newcommand{\stricttransksimx}{\prec^{k\textrm{-}x}}
\newcommand{\stricttranskbwsim}{\prec^{k\textrm{-}\mathsf{bw}}}
\newcommand{\stricttranskdisim}{\prec^{k\textrm{-}\mathsf{di}}}

\newcommand{\stricttranskfsim}{\prec^{k\textrm{-}\mathsf f}}

\newcommand{\cprex}[2]{\mathsf{CPre}^{#1}(#2)}
\newcommand{\cpredi}[1]{\cprex {\mathrm{di}} {#1}}
\newcommand{\cprebw}[1]{\cprex {\mathrm{bw}} {#1}}

\newcommand{\cprelong}[3]{\mathsf{CPre}(#1, #2, #3)}
\newcommand{\cpreone}[2]{\mathsf{CPre}^1(#1, #2)}
\newcommand{\cpretwo}[2]{\mathsf{CPre}^2(#1, #2)}

\newcommand{\st}{\ |\ }

\newcommand{\tickOK}{\checkmark}
\newcommand{\tickNO}{\times}

\tikzset{
	smallstate/.style={state,
		circle, 
		minimum size=4mm, 
		font=\itshape
	},
	initial text=,
	>=stealth',
	>/.style={ultra thick}
}

%% file: introduction.tex
\section{Introduction}

Nondeterministic B\"uchi automata are an effective way to represent
and manipulate $\omega$-regular languages, since they are closed under boolean operations.
They appear in many automata-based formal software verification methods, as well as in
decision procedures for logical theories.
For example, in LTL software model checking 
\cite{Holzmann:Spinbook,optimizing:concur2000}, temporal logic
specifications are converted into B\"uchi automata. 
In other cases, different versions of a program (obtained by abstraction or
refinement of the original) are translated into automata whose languages are then
compared. Testing the conformance of an implementation with its requirements 
specification thus reduces to a language inclusion or language equivalence
problem.
Another application of B\"uchi automata in software engineering is program termination
analysis by the size-change termination method \cite{Lee:SCT2001,seth:buchi}. 
Via an abstraction of the effect of program operations on
data, the termination problem can often be reduced to a
language inclusion problem about two derived B\"uchi automata. 

Our goal is to improve the efficiency and scalability of automata-based 
formal software verification methods.
We consider efficient algorithms for the minimization of automata, in the
sense of obtaining a smaller automaton with the same language, though not
necessarily with the absolute minimal possible number of states. (And, in general,
the minimal automaton for a language is not even unique.) The reason to perform
minimization is that the smaller minimized automaton is more efficient 
to handle in a subsequent computation. Thus there is an algorithmic tradeoff
between the effort for minimization and the complexity of the problem later
considered for this automaton. If only computationally easy questions 
are asked (e.g., reachability/emptiness; solvable in Logspace/PTIME)
then extensive minimization usually does not pay off.
Instead, the main applications are the following:
\begin{enumerate}
\item
Computationally hard automata problems like universality, equivalence, 
and inclusion. These are PSPACE-complete
\cite{kupfermanvardi:fair_verification},
but many practically efficient methods have been developed 
\cite{dill:inclusion:1992,doyen:raskin:antichains,Pit06,antichain:NFA:improved:10,seth:buchi,seth:efficient,abdulla:simulationsubsumption,Rabit_CONCUR2011}.
Still, these all have exponential time complexity and do not scale well.
Typically they are applied to automata with 15--100 states (unless the
automaton has a particularly simple structure).
Thus, one should first minimize the automata before applying these
exponential-time methods. A good minimization algorithm makes it possible
to solve much larger instances.
Even better, many instances of the PSPACE-complete universality, equivalence, 
and inclusion problems can already be solved in the {\em polynomial time} minimization
algorithm (e.g., by reducing the automaton to the trivial universal automaton),
so that the complete {\em exponential time} methods only need to be 
invoked in a small minority of instances.
\item
Cases where the size of an automaton strongly affects the complexity of an
algorithm. In LTL model checking \cite{Holzmann:Spinbook}
one searches for loops in a graph
that is the {\em product} of a large system specification with an automaton
derived from an LTL-formula. Smaller automata often make this easier, though
in practice it also depends on the degree of nondeterminism
\cite{Sebastiani-Tonetta:2003}.
\item
Procedures that combine and modify automata repeatedly.
Model checking algorithms and automata-based decision procedures for
logical theories compute automata products, unions, complements, projections,
etc., and thus the sizes of automata grow rapidly.
Thus, it is important to intermittently minimize the automata to keep their
size manageable, e.g., \cite{Talence:Presburger}.
\end{enumerate}
In general, finding an automaton with the minimal number of states for a
given language is computationally hard; even deciding whether a given
automaton is minimal is already PSPACE-complete \cite{ravikumar:hard:1991}.
Thus much effort has been devoted to finding methods for partial minimization
\cite{optimizing:concur2000,etessami:etal:fairsimulations:05,piterman:generalized06,buchiquotient:ICALP11}.
Simulation preorders played a central role in these efforts, because they
provide PTIME-computable under-approximations of trace inclusions.
However, the quality of the approximation is insufficient in many practical examples.
Multipebble simulations \cite{etessami:hierarchy02} yield coarser
relations by allowing the Duplicator player to hedge her bets in the
simulation game,
but they are not easily computable in practice.
%
\begin{enumerate}
\item
We present methods for transition pruning, i.e., removing transitions from
automata without changing their language. The idea is that certain transitions
can be removed, because other `better' transitions remain. The `better'
criterion relies on combinations of forward and backward simulations and trace
inclusions. We provide a complete picture which combinations are correct to
use for pruning. Moreover, the pruned transitions can be removed `in
parallel' (i.e., without re-computing the simulations and trace inclusions
after every change), which makes the method efficient and practical.
\item 
We present an efficient practical method
to compute good under-approximations of trace inclusions, by introducing \emph{lookahead simulations}.
While it is correct to use full trace inclusions and maximal-pebble
multipebble simulations in our minimization methods, these are not 
easily computed (PSPACE-hard). However, lookahead simulations are PTIME-computable,
and it is correct to use them instead of the more expensive trace inclusions and multipebble simulations.
Lookahead itself is a classic concept in parsing and
many other areas, but it can be defined in many different variants.
Our contribution is to identify and formally describe the lookahead-variant
for simulation preorders that gives the optimal compromise between efficient
computability and maximizing the sizes of the relations.%
\footnote{A thorough literature search showed that this has never been formally
  described so far.}
Practical degrees of lookahead range from 4 to 25, depending on the size and
shape of the automata.
Our experiments show that even moderate lookahead helps considerably in 
obtaining good approximations of trace-inclusions and multipebble
simulations.
\item
We show that variants of the {\em polynomial time} minimization
algorithm can solve most instances of the PSPACE-complete language inclusion
problem. Thus, the complete {\em exponential time} methods of
\cite{dill:inclusion:1992,doyen:raskin:antichains,antichain:NFA:improved:10,seth:buchi,seth:efficient,abdulla:simulationsubsumption,Rabit_CONCUR2011}
need only be invoked in a minority of the cases.
This allows to scale language inclusion testing to much larger instances
(e.g., automata with $\ge 1000$ states) which are beyond traditional methods.
\item
We performed extensive tests of our algorithm on automata of up-to 20000
states. These included random automata according to the
Tabakov-Vardi model \cite{tabakov:model}, 
automata obtained from LTL formulae, and real-world mutual
exclusion protocols. The empirically determined average-case time complexity on
random automata is quadratic, while the (never observed) worst-case complexity
is $O(n^4)$. The worst-case space complexity is quadratic. 
Our algorithm always minimizes better, on average, than all previously
available practical methods. However, the exact advantage varies,
depending on the type of instances; cf. Section~\ref{sec:experiments}.
For example, consider random automata with 100--1000 states, binary alphabet 
and varying transition density ${\it td}$. 
Random automata with ${\it td}=1.4$ cannot be minimized much by {\em any method}.
The only effect is achieved by the trivial removal of dead states which, on
average, yields automata of $78\%$ of the original size.
On the other hand, for ${\it td}=1.8,\dots,2.2$, the best previous minimization methods
yielded automata of $85\%$--$90\%$ of the original size on average, while our algorithm
yielded automata of $3\%$--$15\%$ of the original size on average.
\end{enumerate}
\noindent While we present our methods in the framework of B\"uchi automata, they
directly carry over to the simpler case of finite-word automata.

%% file: preliminaries.tex
\section{Preliminaries}\label{sec:preliminaries}

A \emph{non-deterministic B\"uchi Automaton (BA)} $\A$ is a tuple $(\Sigma, Q, I, F, \delta)$
where $\Sigma$ is a finite alphabet, $Q$ is a finite set of states,
$I \subseteq Q$ is the set of \emph{initial} states,
$F \subseteq Q$ is the set of \emph{accepting} states,
and $\delta \subseteq Q \times \Sigma \times Q$ is the transition relation.
We write $p \goesto \symb q$ for $(p, \symb, q) \in \delta$.
A transition is \emph{transient} iff any path can contain it at most once.
\ignore{
Our automata are non-deterministic in the sense that
\begin{inparaenum}[1)]
	\item there might be several initial states, and
	\item there might be several transitions $p \goesto \symb q$
	for a given state $p$ and input symbol $\symb$ (and several initial states).
\end{inparaenum}
}
To simplify the presentation, we assume that automata are \emph{forward and backward complete}, i.e.,
for any state $p \in Q$ and symbol $\symb \in \Sigma$,
there exist states $q_0, q_1 \in Q$ s.t. $q_0 \goesto \symb p \goesto \symb q_1$.
Every automaton can be converted into an equivalent complete one by adding
at most two states and a linear number of transitions.%
\footnote{
For efficiency reasons, our implementation works directly on incomplete automata.}
A state is \emph{dead} iff either it is not reachable from an initial state,
or it cannot reach an accepting loop.
In our simplification techniques, we always remove dead states.

A B\"uchi automaton $\A$ describes a set of infinite words (its language), i.e., a subset of $\Sigma^\omega$.
An \emph{infinite trace} of $\A$ on a word $w = \symb_0\symb_1 \cdots \in \Sigma^\omega$ (or \emph{$w$-trace})
\emph{starting} in a state $q_0 \in Q$ is an infinite sequence of transitions 
$\pi = q_0 \goesto {\symb_0} q_1 \goesto {\symb_1} \cdots$.
By $\prefix \pi i$ we denote the finite prefix $\pi = q_0 \goesto {\symb_0} \cdots \goesto {\symb_{i-1}} q_i$,
and by $\suffix \pi i$ the infinite suffix $q_i \goesto {\symb_i} q_{i+1} \goesto {\symb_{i+1}} \cdots$.
Finite traces starting in $q_0$ and \emph{ending} in a state $q_m \in Q$ are defined similarly.
A finite or infinite trace is \emph{initial} iff it starts in an initial state $q_0 \in I$;
if it is infinite, then it is \emph{fair} iff $q_i \in F$ for infinitely many $i$.
The \emph{language of $\A$} is $\lang \A = \{w \in \Sigma^\omega \mid \mbox{$\A$ has an infinite, initial and fair trace on $w$} \}$.

\xparagraph{Language inclusion.}
When automata are viewed as a finite representation for languages,
it is natural to ask whether two different automata represent the same language,
or, more generally, to compare these languages for inclusion.
Formally, for two automata $\A = (\Sigma, Q_\A, I_\A, F_\A, \delta_\A)$ and $\B = (\Sigma, Q_\B, I_\B, F_\B, \delta_\B)$
we write $\A \languageinclusion \B$ iff $\lang \A \subseteq \lang \B$ and $\A \languageequivalence \B$ iff $\lang \A = \lang \B$.
The \emph{language inclusion/equivalence problem} consists in determining whether $\A \languageinclusion \B$ or $\A \languageequivalence \B$ holds, respectively.
For general non-deterministic automata,
language inclusion and equivalence are PSPACE-complete \cite{kupfermanvardi:fair_verification}
(which entails that, under standard theoretic-complexity assumptions, they admit no efficient deterministic algorithm).
Therefore, one considers suitable under-approximations.
%
\begin{definition}
	A preorder $\sqsubseteq$ on $Q_\A \times Q_\B$ is \emph{good for inclusion} (GFI)
	iff the following holds: If $\forall q\in I_\A \exists q' \in I_\B \cdot q \sqsubseteq q'$, then $\A \languageinclusion \B$.
\end{definition}
In other words, GFI preorders give a sufficient condition for inclusion,
by matching initial states of $\A$ with initial states of $\B$.
(They are not necessary for inclusion since there are several initial states.)
Moreover, if computing a GFI preorder is efficient, than also inclusion can be established efficiently.
Finally, if a preorder is GFI, then all smaller preorders are GFI too, i.e., GFI is $\subseteq$-downward closed.

\xparagraph{Quotienting.}

Another interesting problem is how to simplify an automaton while preserving
its semantics, i.e., its language. Generally, one tries to reduce the number
of states/transitions. 
This is useful because the complexity of decision procedures usually depends
on the size of the input automata.

A classical operation for reducing the number of states of an automaton is that of quotienting,
where states of the automaton are identified according to a given equivalence, and transitions are projected accordingly.
Since in practice we obtain quotienting equivalences from suitable preorders,
we directly define quotienting w.r.t. a preorder.
Formally, fix a BA $\A = (\Sigma, Q, I, F, \delta)$ and a preorder $\sqsubseteq$ on $Q$,
with induced equivalence $\equiv = \sqsubseteq \cap \sqsupseteq$.
Given a state $q \in Q$, we denote by $[q]$ its equivalence class w.r.t. $\equiv$,
and, for a set of states $P \subseteq Q$, $[P]$ is the set of equivalence classes $[P] = \{ [p] \st p \in P \}$.
\begin{definition}
	The \emph{quotient} of $\A$ by $\sqsubseteq$ is
	$\A/\!\sqsubseteq = (\Sigma, [Q], [I], [F], \delta')$,
	where $\delta' = \{([q_1],\symb,[q_2]) \st \exists q_1' \in [q_1], q_2' \in [q_2].\, (q_1',\symb,q_2') \in \delta\}$,
	i.e., transitions are induced element-wise.
\end{definition}
Clearly, every trace $q_0 \goesto {\symb_0} q_1 \goesto {\symb_1} \cdots$ in $\A$
immediately induces a corresponding trace $[q_0] \goesto {\symb_0} [q_1] \goesto {\symb_1} \cdots$ in $\A/\!\sqsubseteq$,
which is fair/initial if the former is fair/initial, respectively.
Consequently, $\A \languageinclusion \A/\!\sqsubseteq$ for \emph{any} preorder $\sqsubseteq$.
If, additionally, $\A/\!\sqsubseteq\ \languageinclusion\ \A$,
then we say that the preorder $\sqsubseteq$ is good for quotienting (GFQ).
\begin{definition}
	A preorder $\sqsubseteq$ is \emph{good for quotienting} iff
	$\A/\!\sqsubseteq\ \languageequivalence\ \A$.
\end{definition}
%
%
Like GFI preorders, also GFQ preorders are downward closed (since a smaller preorder is quotienting ``less'').
Therefore, we are interested in efficiently computable GFI/GFQ preorders.
A classical example is given by \emph{simulation relations}.

\xparagraph{Simulation relations.}
\ignore{
	Trace inclusions are PSPACE-complete to compute in general \cite{some paper by Vardi};
	intuitively, this holds since the quantification pattern ``$\forall\exists$'' is over arbitrary long traces.
	By restricting the quantification to only 1-step,
	and by building the trace incrementally by repeated quantification ``$\forall\exists\forall\exists\cdots$'',
	one obtains \emph{simulation relations} \cite{milner and park}, which are computable in PTIME.
}
Basic forward simulation is a binary relation on the states of $\A$;
it relates states whose behaviors are step-wise related,
which allows one to reason about the internal structure of automaton $\A$---
i.e., \emph{how} a word is accepted, and not just \emph{whether} it is accepted.
Formally, simulation between two states $p_0$ and $q_0$ can be described in terms of a game between two players, Spoiler and Duplicator,
where the latter wants to prove that $q_0$ can step-wise mimic any behavior of $p_0$, and the former wants to disprove it.
The game starts in the initial configuration $(p_0, q_0)$.
Inductively, given a game configuration $(p_i, q_i)$ at the $i$-th round of the game,
Spoiler chooses a symbol $\symb_i \in \Sigma$ and a transition $\trans{p_i}{\symb_i}{p_{i+1}}$.
Then, Duplicator responds by choosing a matching transition $\trans{q_i}{\symb_i}{q_{i+1}}$,
and the next configuration is $(p_{i+1}, q_{i+1})$.
Since the automaton is assumed to be complete, the game goes on forever,
and the two players build two infinite traces
$\pi_0 = p_0 \goesto {\symb_0} p_1 \goesto {\symb_1} \cdots$ and $\pi_1 = q_0 \goesto {\symb_0} q_1 \goesto {\symb_1} \cdots$.
The winning condition depends on the type of simulation,
and different types have been considered depending on whether one is interested in GFQ or GFI relations.
Here, we consider \emph{direct} \cite{dill:inclusion:1992},
\emph{delayed} \cite{etessami:etal:fairsimulations:05}
and \emph{fair simulation} \cite{fairsimulation:02}.
Let $x \in \{\mathrm{di, de, f}\}$.
Duplicator wins the play if $\mathcal C^x(\pi_0, \pi_1)$ holds, where 
\begin{align}
	\mathcal C^{\mathrm {di}}(\pi_0, \pi_1)	&\iff \forall (i \geq 0) \cdot p_i \in F \implies q_i \in F \\
	\mathcal C^{\mathrm {de}}(\pi_0, \pi_1)	&\iff \forall (i \geq 0) \cdot p_i \in F \implies \exists (j \geq i) \cdot q_j \in F \\
	\mathcal C^{\mathrm f}(\pi_0, \pi_1)	&\iff \textrm{ if $\pi_0$ is fair, then $\pi_1$ is fair }
\end{align}
Intuitively, direct simulation requires that accepting states are matched immediately (the strongest condition),
while in delayed simulation Duplicator is allowed to accept only after a finite delay.
In fair simulation (the weakest condition),
Duplicator must visit accepting states only if Spoiler visits infinitely many of them.
Thus, $\mathcal C^{\mathrm {di}}(\pi_0, \pi_1)$ implies $\mathcal C^{\mathrm {de}}(\pi_0, \pi_1)$,
which, in turn, implies $\mathcal C^{\mathrm f}(\pi_0, \pi_1)$.

We define $x$-simulation relation $\xsim x \subseteq Q \times Q$
by stipulating that $p_0 \xsim x q_0$ iff Duplicator has a winning strategy in the $x$-simulation game,
starting from configuration $(p_0, q_0)$;
clearly, $\disim \subseteq \desim \subseteq \fsim$.
Simulation between states in different automata $\A$ and $\B$ can be computed as a simulation on their disjoint union.
All these simulation relations are GFI preorders which can be computed in polynomial time
\cite{dill:inclusion:1992,HHK:FOCS95,etessami:etal:fairsimulations:05};
moreover, direct and delayed simulation are GFQ \cite{etessami:etal:fairsimulations:05},
but fair simulation is not \cite{fairsimulation:02}. 
\begin{lemma}[\cite{dill:inclusion:1992,HHK:FOCS95,fairsimulation:02,etessami:etal:fairsimulations:05}]
	For $x \in \{ \mathrm{di, de, f} \}$, $x$-simulation $\sqsubseteq^x$ is a PTIME, GFI preorder,
	and, for $y \in \{ \mathrm{di, de} \}$, $\sqsubseteq^y$ is also GFQ.
\end{lemma}
%

\xparagraph{Trace inclusions.}

While simulations are efficiently computable,
their use is often limited by their size,
which can be much smaller than other GFI/GFQ preorders.
One such example of coarser GFI/GFQ preorders is given by \emph{trace inclusions},
which are obtained through a modification of the simulation game, as follows.

In simulation games, the players build two paths $\pi_0, \pi_1$ by choosing single transitions in an alternating fashion;
Duplicator moves by knowing only the next 1-step move of Spoiler.
We can obtain coarser relations by allowing Duplicator a certain amount of \emph{lookahead} on Spoiler's moves.
In the extremal case of $\omega$-lookahead, i.e.,
where Spoiler has to reveal her whole path in advance,
we obtain trace inclusions.

Analogously to simulations, we define direct, delayed, and fair trace inclusion, as binary relations on $Q$.
For $x \in \{\mathrm{di, de, f}\}$, \emph{$x$-trace inclusion} holds between $p$ and $q$, written $p \xincl x q$ iff,
for every word $w = \symb_0\symb_1 \cdots \in \Sigma^\omega$,
and for every infinite $w$-trace $\pi_0 = p_0 \goesto {\symb_0} p_1 \goesto {\symb_1} \cdots$ starting at $p_0 = p$, 
there exists an infinite $w$-trace $\pi_1 = q_0 \goesto {\symb_0} q_1 \goesto {\symb_1} \cdots$ starting at $q_0 = q$, 
s.t. $\mathcal C^x(\pi_0, \pi_1)$.
All these trace inclusions are GFI preorders subsuming the corresponding simulation, i.e., $\xsim x\ \subseteq \xincl x$
(since Duplicator has more power in the trace inclusion game);
also, $\directtraceinclusion$ is a subset of $\delayedtraceinclusion$, which, in turn, is a subset of $\fairtraceinclusion$.
%
%
\ignore{
Note that fair trace inclusion is different from language inclusion,
the former being a relation between states, while the latter is a relation between automata.
However, trace inclusions are GFI, and thus they can be used to establish inclusion between automata
(If both automata have exactly one initial state,
then fair trace inclusion is even complete for language inclusion,
i.e., fair trace inclusion between the unique initial states holds iff language inclusion holds.)
}
Regarding quotienting, $\directtraceinclusion$ is GFQ (like $\disim$; this follows from \cite{etessami:hierarchy02}),
while $\fairtraceinclusion$ is not, since it is coarser than fair simulation, which is not GFQ \cite{fairsimulation:02}.
While delayed simulation $\desim$ is GFQ,
delayed trace inclusion $\delayedtraceinclusion$ is not GFQ \cite{buchiquotient:ICALP11}.
\begin{lemma}
	For $x \in \{ \mathrm{di, de, f} \}$, $x$-trace inclusion $\subseteq^x$ is a GFI preorder.
	Moreover, $\directtraceinclusion$ is a GFQ preorder.
\end{lemma}
\noindent
Finally, though $\desim$ and $\directtraceinclusion$ are incomparable,
there exists a common generalization included in $\delayedtraceinclusion$
called \emph{delayed fixed-word simulation} which is GFQ \cite{buchiquotient:ICALP11}.%
\footnote{Delayed fixed-word simulation is defined as a variant of simulation
where Duplicator has $\omega$-lookahead only on the input word $w$, and \emph{not} on Spoiler's actual $w$-trace $\pi_0$;
that it subsumes $\directtraceinclusion$ is non-trivial.}

\ignore{
	\emph{Delayed fixed-word simulation $\delayedfixedwordsimulation$} \cite{ICALP11} can be defined as a variant of simulation,
	where Duplicator has $\omega$-lookahead only on the input word $w = \symb_0\symb_1 \cdots \in \Sigma^\omega$,
	but not on Spoiler's trace $\pi_0$: Once Spoiler reveals $w$,
	the two players alternate in choosing transitions like in ordinary simulation,
	with the proviso that the input symbol at round $i$ is $\symb_i$.
	The winning condition on the resulting paths $\pi_0, \pi_1$ is $\mathcal C^{\mathrm {de}}(\pi_0, \pi_1)$, like in delayed simulation.
	Delayed fixed-word simulation is a GFQ preorder;
	by definition, it is coarser than delayed simulation (strictly coarser, in fact),
	and, more importantly, it can be proved to strictly subsume $\directtraceinclusion$ \cite{ICALP11}.
}

\xparagraph{Backward simulation and trace inclusion.}

Yet another way of obtaining GFQ/GFI preorders is to consider variants of simulation/trace inclusion which go backwards in time.
\emph{Backward simulation} $\bwsim$ (\cite{somenzi:efficient}, where it is called \emph{reverse simulation})
is defined like ordinary simulation, except that transitions are taken backwards:
From configuration $(p_i, q_i)$, Spoiler selects a transition $\trans {p_{i+1}} {\symb_i} {p_i}$,
Duplicator replies with a transition $\trans {q_{i+1}} {\symb_i} {q_i}$,
and the next configuration is $(p_{i+1}, q_{i+1})$.
Let $\pi_0$ and $\pi_1$ be the two infinite backward traces built in this way.
The corresponding winning condition considers both accepting and initial states:
\begin{align}
	\mathcal C^\mathrm{bw}(\pi_0, \pi_1) \ \iff\ \forall (i \geq 0) \cdot
	\left\{\begin{array}{l}
		p_i \in F \implies q_i \in F, \textrm{ and } \\
		p_i \in I \implies q_i \in I
	\end{array}\right.
\end{align}
$\bwsim$ is an efficiently computable GFQ preorder \cite{somenzi:efficient} incomparable with forward simulations.
It can be used to establish language inclusion by matching final states of $\A$ with final states of $\B$
(dually to forward simulations); in this sense, it is GFI.
\begin{lemma}[\cite{somenzi:efficient}]
	Backward sim. is a PTIME GFQ/GFI preorder.
\end{lemma}
\noindent
The corresponding notion of \emph{backward trace inclusion} $\bwdirecttraceinclusion$ is defined as follows:
$p \bwdirecttraceinclusion q$ iff,
for every finite word $w = \symb_0\symb_1 \cdots \symb_{m-1} \in \Sigma^*$,
and for every initial, finite $w$-trace
$\pi_0 = p_0 \goesto {\symb_0} p_1 \goesto {\symb_1} \cdots \goesto {\symb_{m-1}} p_m$ ending in $p_m = p$,
there exists an initial, finite $w$-trace
$\pi_1 = q_0 \goesto {\symb_0} q_1 \goesto {\symb_1} \cdots \goesto {\symb_{m-1}} q_m$ ending in $q_m = q$,
s.t., for any $i \geq 0$,  if $p_i \in F$, then $q_i \in F$.
Note that backward trace inclusion deals with \emph{finite traces} (unlike forward trace inclusions),
which is due to the asymmetry between past and future in $\omega$-automata.
Clearly, $\bwsim \subseteq \bwdirecttraceinclusion$;
we observe that even $\bwdirecttraceinclusion$ is GFQ/GFI.
\begin{theorem}\label{lem:bwincl_GFQ}
	Backward trace inclusion \ignore{$\bwdirecttraceinclusion$} is a GFQ/GFI preorder.%
\end{theorem}
\begin{proof}

	We first prove that $\bwdirecttraceinclusion$ is GFQ.
	Let $\sqsubseteq := \bwdirecttraceinclusion$.
	Let $w = \symb_0 \symb_1 \cdots \in \lang{\A/\!\sqsubseteq}$, and we show $w \in \lang \A$.
	There exists an initial, infinite and fair $w$-trace
	$\pi = [q_0] \goesto {\symb_0} [q_1] \goesto {\symb_1} \cdots$.
	For $i \geq 0$, let $w_i = \symb_0 \symb_1 \cdots \symb_i$ (with $w_{-1} = \varepsilon$),
	and let $\pi[0..i]$ be the $w_{i-1}$-trace prefix of $\pi$.
	For any $i \geq 0$, we build by induction an initial, finite $w_{i-1}$-trace $\pi_i$ ending in $q_i$ (of length $i$)
	visiting at least as many accepting states as $\pi[0..i]$ (and at the same time $\pi[0..i]$ does).
	
	For $i = 0$, just take the empty $\varepsilon$-trace $\pi_0 = q_0$.
	For $i > 0$, assume that an initial $w_{i-2}$-trace $\pi_{i-1}$ ending in $q_{i-1}$ has already been built.
	We have the transition $\trans {[q_{i-1}]} {\symb_{i-1}} {[q_i]}$ in $\lang{\A/\!\sqsubseteq}$.
	There exist $\hat q \in [q_{i-1}]$ and $\hat q' \in [q_i]$ s.t.
	we have a transition $\trans {\hat q} {\symb_{i-1}} {\hat q'}$ in $\A$.
	W.l.o.g. we can assume that $\hat q' = q_i$, since $[q_i] = [\hat q']$.
	By $q_{i-1} \bwdirecttraceinclusion \hat q$, there exists an initial, finite $w_{i-2}$-trace $\pi'$ ending in $\hat q$.
	By the definition of backward inclusion, $\pi'$ visits at least as many accepting states as $\pi_{i-1}$,
	which, by inductive hypothesis, visits at least as many accepting states as $\pi[0..i-1]$.
	Therefore, $\pi_i := \pi' \goesto {\symb_{i-1}} q_i$ is an initial, finite $w_{i-1}$-trace ending in $q_i$.
	Moreover, if $[q_i] \in F'$, then, since backward inclusion respects accepting states, $[q_i] \subseteq F$,
	hence $q_i \in F$, and, consequently, $\pi_i$ visits at least as many accepting states as $\pi[0..i]$.
	Since $\pi$ is fair, the finite, initial traces $\pi_0, \pi_1, \cdots$ visit unboundedly many accepting states.
	Since $\A$ is finitely branching, by K\"onig's Lemma there exists an initial, infinite and fair $w$-trace $\pi_\omega$.
	Therefore, $w \in \lang \A$.
	
	We now prove that $\bwdirecttraceinclusion$ is GFI.
	Let $\A$ and $\B$ be two automata.
	For backward notions, we require that every accepting state in $\A$ is in relation with an accepting state in $\B$.
	Let $w = \symb_0 \symb_1 \cdots \in \lang \A$,
	and let $\pi_0 = p_0 \goesto {\symb_0} p_1 \goesto {\symb_0} \cdots$ be an initial and fair $w$-path in $\A$.
	Since $\pi_0$ visits infinitely many accepting states,
	and since each such state is $\bwdirecttraceinclusion$-related to an accepting state in $\B$,
	by using the definition of $\bwdirecttraceinclusion$
	it is possible to build in $\B$ longer and longer finite initial traces in $\B$
	visiting unboundedly many accepting states.
	Since $\B$ is finitely branching,
	by K\"onig's Lemma there exists an infinite, initial and fair $w$-trace $\pi_\omega$ in $\B$.
	Thus, $w \in \lang \B$.
\end{proof}

%% file: minimization.tex
\section{Transition Pruning Minimization Techniques} \label{sec:minimization}

While quotienting-based minimization techniques reduce the number of states by merging them,
we explore an alternative method which prunes (i.e., removes) transitions.
The intuition is that certain transitions can be removed from an automaton without changing its language
when other `better' transitions remain. 
\begin{definition}
	Let $\A = (\Sigma, Q, I, F, \delta)$ be a BA and let $\prunerel$ a transitive, asymmetric relation on $\delta$.
	The \emph{pruned automaton} is defined as $\prune{\A}{\prunerel} := (\Sigma, Q, I, F, \delta')$,
	with $\delta' = \{(p,\symb,r) \in \delta \st \nexists (p',\symb',r') \in \delta \cdot (p,\symb,r) \prunerel (p',\symb',r')\}$.
\end{definition}
%
\noindent
By the assumptions on $\prunerel$, the pruned automaton
$\prune{\A}{\prunerel}$ is uniquely defined.
Notice that transitions are removed `in parallel'. 
Though $\prunerel$ might depend on $\delta$,
$\prunerel$ is {\em not} re-computed even if the removal of a single transition changes
$\delta$. This is important because computing $\prunerel$ may be expensive.
Since removing transitions cannot introduce new words in the language,
$\prune{\A}{\prunerel} \languageinclusion \A$.
When also the converse inclusion holds (so the language is preserved),
we say that $\prunerel$ is \emph{good for pruning} (GFP),
i.e., $\prunerel$ is GFP iff $\prune{\A}{\prunerel} \!\languageequivalence\! \A$.
Clearly, GFP is $\subseteq$-downward closed (like GFI and GFQ).

We study GFP relations obtained by comparing the endpoints of
transitions over the same input symbol.
Formally, given two binary relations $\brel, \frel \subseteq Q \times Q$, we define
\[ \makeprunerel{\brel}{\frel} = \{((p,\symb,r),(p',\symb,r')) \st p \brel p' \textrm{ and } r \frel r' \} \]
$\makeprunerel{\cdot}{\cdot}$ is monotone in both arguments.
In the following, we explore which state relations $\brel,\frel$ induce GFP relations $\makeprunerel{\brel}{\frel}$.

It has long been known that 
$\makeprunerel{\id}{\strictdisim}$ and $\makeprunerel{\strictbwsim}{\id}$ are GFP
(see \cite{simulationminimization:03} where the removed transitions are called `little brothers').
Moreover, even the relation 
$R_t(\strictlanguageinclusion) := \makeprunerel{\id}{\strictdisim} \cup 
\{((p,\symb,r),(p,\symb,r')) \st \mbox{$(p,\symb,r')$ is transient and}
\ r \strictlanguageinclusion r'\}$ is GFP \cite{somenzi:efficient}, i.e.,
strict fair trace inclusion
suffices if the remaining transition can only be used once.
However, in general, $\makeprunerel{\id}{\strictlanguageinclusion}$ is not GFP.
Moreover, even if only transient transitions 
are compared/pruned, 
$\makeprunerel{\strictbwsim}{\strictlanguageinclusion}$ is not GFP; cf. Fig.~\ref{fig:pruning}.

\input{GFP_table}

\begin{theorem}\label{lem:prune_id_strictdirecttraceinclusion}
For every asymmetric and transitive relation $R \subseteq \directtraceinclusion$,
$\makeprunerel{\id}{R}$ is GFP.
\end{theorem}
\begin{proof}
Let $\A' = \prune{\A}{\makeprunerel{\id}{R}}$. 
We show $\A \languageinclusion \A'$.
If $w = \symb_0\symb_1 \cdots\in \lang{\A}$ then there exists an infinite fair
initial trace $\hat{\pi}$ on $w$ in $\A$. We show $w \in \lang{\A'}$.

We call a trace $\pi = q_0 \goesto {\symb_0} q_1 \goesto {\symb_1} \cdots$ on
$w$ in $\A$ $i$-{\em good} if it does not contain any
transition $q_j \goesto {\symb_j} q_{j+1}$ for $j < i$ s.t. there exists an $\A$ transition
$q_j \goesto {\symb_j} q_{j+1}'$ with $q_{j+1}\, R\, q_{j+1}'$ (i.e., no such
transition is used within the first $i$ steps).
Since $\A$ is finitely branching, for every state and symbol there exists at
least one $R$-maximal successor that is still present in
$\A'$, because $R$ is asymmetric and transitive.
Thus, for every $i$-good trace $\pi$ on
$w$ there exists an $(i+1)$-good trace $\pi'$ on $w$ s.t.
$\pi$ and $\pi'$ are identical on the first $i$ steps and 
$\mathcal C^{\mathrm {di}}(\pi, \pi')$, because
$R \subseteq \directtraceinclusion$.
Since $\hat{\pi}$ is an infinite fair initial trace on $w$ (which is trivially
$0$-good), there exists an
infinite initial trace $\tilde{\pi}$ on $w$ that is
$i$-good for every $i$ and $\mathcal C^{\mathrm{di}}(\hat{\pi}, \tilde{\pi})$.
Moreover, $\tilde{\pi}$ is a trace in $\A'$.
Since $\hat{\pi}$ is fair and $\mathcal C^{\mathrm{di}}(\hat{\pi}, \tilde{\pi})$,  
$\tilde{\pi}$ is an infinite fair initial trace on $w$ that is $i$-good for every $i$.
Therefore $\tilde{\pi}$ is a fair initial trace on $w$ in $\A'$ and thus $w \in \lang{\A'}$.
\end{proof}

\begin{theorem}\label{lem:prune_strictbwdirecttraceinclusion_id}
For every asymmetric and transitive relation $R \subseteq \bwdirecttraceinclusion$,
$\makeprunerel{R}{\id}$ is GFP.
\end{theorem}
\begin{proof}
Let $\A' = \prune{\A}{\makeprunerel{R}{\id}}$. We show $\A \languageinclusion \A'$.
If $w = \symb_0\symb_1 \cdots\in \lang{\A}$ then there exists an infinite fair
initial trace $\hat{\pi}$ on $w$ in $\A$. We show $w \in \lang{\A'}$.

We call a trace 
$\pi = q_0 \goesto {\symb_0} q_1 \goesto {\symb_1} \cdots$ on $w$ in $\A$
$i$-{\em good} if it does not contain any
transition $q_j \goesto {\symb_j} q_{j+1}$ for $j < i$ s.t. there exists an $\A$ transition
$q_j' \goesto {\symb_j} q_{j+1}$ with $q_j\, R\, q_j'$ (i.e., no such
transition is used within the first $i$ steps).

We show, by induction on $i$, the following property (P):
For every $i$ and every initial trace $\pi$ on $w$ in $\A$ there exists an
initial $i$-good trace $\pi'$ on $w$ in $\A$ s.t.
$\pi$ and $\pi'$ have identical suffixes from step $i$ onwards
and $\mathcal C^{\mathrm {di}}(\pi, \pi')$.

The base case $i=0$ is trivial with $\pi'=\pi$.
For the induction step there are two cases. If $\pi$ is $(i+1)$-good
then we can take $\pi' = \pi$.
Otherwise there exists a transition $q_i' \goesto {\symb_i} q_{i+1}$ with
$q_i\, R\, q_i'$. Without restriction (since $\A$
is finite and $R$ is asymmetric and transitive) we assume
that $q_i'$ is $R$-maximal among the
$\symb_i$-predecessors of $q_{i+1}$.
In particular, the transition $q_i' \goesto {\symb_i} q_{i+1}$
is present in $\A'$.
Since $R \subseteq \bwdirecttraceinclusion$, there exists
an initial trace $\pi''$ on $w$ that has suffix
$q_i' \goesto {\symb_i} q_{i+1} \goesto{\symb_{i+1}} q_{i+2} \dots$
and $\mathcal C^{\mathrm {di}}(\pi, \pi'')$.
Then, by induction hypothesis, there exists an initial $i$-good
trace $\pi'$ on $w$ that has suffix
$q_i' \goesto {\symb_i} q_{i+1} \goesto{\symb_{i+1}} q_{i+2} \dots$
and $\mathcal C^{\mathrm {di}}(\pi'', \pi')$.
Since $q_i'$ is $R$-maximal among the
$\symb_i$-predecessors of $q_{i+1}$ we obtain that $\pi'$ is also 
$(i+1)$-good. Moreover, $\pi'$ and $\pi$ have identical suffixes
from step $i+1$ onwards. Finally, 
by $\mathcal C^{\mathrm {di}}(\pi, \pi'')$
and $\mathcal C^{\mathrm {di}}(\pi'', \pi')$, we obtain
$\mathcal C^{\mathrm {di}}(\pi, \pi')$.

Given the infinite fair initial trace $\hat{\pi}$ on $w$ in $\A$,
it follows from property (P) and K\"onig's Lemma that there
exists an infinite initial trace $\tilde{\pi}$ on $w$ that is
$i$-good for every $i$ and $\mathcal C^{\mathrm {di}}(\hat{\pi},
\tilde{\pi})$.
Therefore $\tilde{\pi}$ is an infinite fair initial trace 
on $w$ in $\A'$ and thus $w \in \lang{\A'}$.
\end{proof}

\begin{theorem}\label{lem:bwsim-fwtrace}
If $\A = \A/\!\bwsim$ then
$\makeprunerel{\strictbwsim}{\directtraceinclusion}$ is GFP.
\end{theorem}
\begin{proof}
Let $\A' = \prune{\A}{\makeprunerel{\strictbwsim}{\directtraceinclusion}}$. We show $\A \languageinclusion \A'$.
Let $w = \symb_0\symb_1 \cdots\in \lang{\A}$. Then there exists an infinite fair
initial trace $\hat{\pi}$ on $w$ in $\A$. We show $w \in \lang{\A'}$.

We call a trace $\pi = q_0 \goesto {\symb_0} q_1 \goesto {\symb_1} \cdots$ 
on $w$ start-maximal iff it is initial and there does not exist any trace
$\pi' = q_0' \goesto {\symb_0} q_1' \goesto {\symb_1} \cdots$ on $w$
s.t. $\mathcal C^{\mathrm {di}}(\pi, \pi')$ and $q_0 \strictbwsim q_0'$.
We call a trace $\pi = q_0 \goesto {\symb_0} q_1 \goesto {\symb_1} \cdots$ 
on $w$ $i$-{\em good} iff it is start-maximal and $\pi$ does not contain any
transition $q_j \goesto {\symb_j} q_{j+1}$ for $j < i$ s.t. there exists an $\A$ transition
$q_j \goesto {\symb_j} q_{j+1}'$ with $q_{j+1}\, \strictbwsim\, q_{j+1}'$ 
and there exists an infinite trace $\suffix {\pi'} {j+1}$ from $q_{j+1}'$ with
$\mathcal C^{\mathrm {di}}(\suffix \pi {j+1}, \suffix {\pi'} {j+1})$.

Since $\A$ is finite, there are $\strictbwsim$-maximal
elements among those finitely many successors of every state $q_j$
from which there exists an infinite trace $\suffix {\pi'} {j+1}$
with $\mathcal C^{\mathrm {di}}(\suffix \pi {j+1}, \suffix {\pi'} {j+1})$.
Thus, for every infinite $i$-good trace $\pi$ on
$w$ there exists an $(i+1)$-good trace $\pi'$ on $w$ s.t.
$\pi$ and $\pi'$ are identical on the first $i$ steps and 
$\mathcal C^{\mathrm {di}}(\pi, \pi')$.

Since there is an infinite fair initial trace $\hat{\pi}$ on $w$, 
there also exists a start-maximal, and thus $0$-good,
fair initial trace on $w$, because $\strictbwsim$ has maximal elements. 
Then it follows from the
property above that there exists an infinite initial 
trace $\tilde{\pi}$ on $w$ that is
$i$-good for every $i$ and $\mathcal C^{\mathrm {di}}(\hat{\pi},
\tilde{\pi})$. In particular, this implies that $\tilde{\pi}$ is fair. 
So $\tilde{\pi}$ is an infinite fair initial trace on $w$ that is $i$-good for every $i$.

Let now $\tilde{\pi} = q_0 \goesto {\symb_0} q_1 \goesto {\symb_1} \cdots$.
We show that $\tilde{\pi}$ is also possible in $\A'$ by assuming the opposite and
deriving a contradiction. Suppose that $\tilde{\pi}$ contains a transition
$q_j \goesto {\symb_j} q_{j+1}$ that is not present in $\A'$.
Then there must exist a transition $q_j' \goesto {\symb_j} q_{j+1}'$ in $\A'$
s.t. $q_j\, \strictbwsim\, q_j'$ and 
$q_{j+1}\, \directtraceinclusion\, q_{j+1}'$.
We cannot have $j=0$, because in this case $\tilde{\pi}$ would not be start-maximal
and thus not even $1$-good. So we get $j \ge 1$.
Since $q_j\, \strictbwsim\, q_j'$ and $q_{j-1} \goesto {\symb_{j-1}} q_j$
there must exist a state $q_{j-1}'$ s.t. $q_{j-1}' \goesto {\symb_{j-1}} q_j'$
and $q_{j-1}\, \bwsim\, q_{j-1}'$. In particular, 
$q_x \in F \,\Rightarrow\, q_{x+1}' \in F$ for $x \in\{j-1, j\}$.
By $\A = \A/\!\bwsim$ we obtain
that either $q_{j-1} = q_{j-1}'$ or $q_{j-1}\, \strictbwsim\, q_{j-1}'$.
The first case would imply that $\pi'$ is not $j$-good, because 
$q_{j+1}\, \directtraceinclusion\, q_{j+1}'$, and thus yield a contradiction.
Therefore, we must have $q_{j-1}\, \strictbwsim\, q_{j-1}'$.
We cannot have $j-1=0$, because in this case $\pi'$ would not be start-maximal
and thus not even $1$-good. So we get $j-1 \ge 1$.
The whole argument above repeats with $j-1, j-2, j-3,\dots$ substituted for
$j$ until we get a contradiction or $0$ is reached. Reaching $0$ also yields a 
contradiction to start-maximality of $\tilde{\pi}$, as above. 
Therefore $\tilde{\pi}$ is a fair initial trace on $w$ in $\A'$ 
and thus $w \in \lang{\A'}$.
\end{proof}

\begin{theorem}\label{lem:bwtrace-fwsim}
$\makeprunerel{\bwdirecttraceinclusion}{\strictdisim}$ is GFP.
\end{theorem}
\begin{proof}
Let $\A' = \prune{\A}{\makeprunerel{\bwdirecttraceinclusion}{\strictdisim}}$. We show $\A \languageinclusion \A'$.
Let $w = \symb_0\symb_1 \cdots\in \lang{\A}$. 
Then there exists an infinite fair
initial trace $\hat{\pi}$ on $w$ in $\A$. We show $w \in \lang{\A'}$.

Given some infinite initial trace $\pi = q_0 \goesto {\symb_0} q_1 \goesto {\symb_1} \cdots$ 
on $w$, we call it $i$-good iff its first $i$ transitions are also possible in
$\A'$.

We now show, by induction on $i$, the following property (P):
For every infinite initial trace $\pi = q_0 \goesto {\symb_0} q_1 \goesto {\symb_1} \cdots$ 
on $w$ and every $i \ge 0$, there exists an infinite initial trace 
$\pi' = q_0' \goesto {\symb_0} q_1' \goesto {\symb_1} \cdots$ 
on $w$ that is $i$-good and $\mathcal C^{\mathrm {di}}(\pi, \pi')$ and
$\forall j \ge i.\, q_j \disim q_j'$.

The base case $i=0$ is trivially true with $\pi' = \pi$.
For the induction step consider an infinite initial trace 
$\pi = q_0 \goesto {\symb_0} q_1 \goesto {\symb_1} \cdots$ on $w$.
By induction hypothesis, there exists an infinite initial trace 
$\pi^1 = q_0^1 \goesto {\symb_0} q_1^1 \goesto {\symb_1} \cdots$ 
on $w$ that is $i$-good and $\mathcal C^{\mathrm {di}}(\pi, \pi^1)$ and
$\forall j \ge i.\, q_j \disim q_j^1$.

If $\pi^1$ is $(i+1)$-good then we are done. Otherwise,
the transition $q_i^1 \goesto {\symb_i} q_{i+1}^1$ is not present in $\A'$.
Since $\A' = \prune{\A}{\makeprunerel{\bwdirecttraceinclusion}{\strictdisim}}$,
there must exist a transition $q_i^2 \goesto {\symb_i} q_{i+1}^2$ in $\A'$
s.t. $q_i^1 \bwdirecttraceinclusion q_i^2$ and 
$q_{i+1}^1 \strictdisim q_{i+1}^2$. 
It follows from the definitions of $\bwdirecttraceinclusion$ and
$\strictdisim$ that there exists an infinite initial trace 
$\pi^2 = q_0^2 \goesto {\symb_0} q_1^2 \goesto {\symb_1} \cdots$ 
on $w$ s.t. $\mathcal C^{\mathrm {di}}(\pi^1, \pi^2)$, 
$q_{i+1}^1 \strictdisim q_{i+1}^2$
and $\forall j \ge i+1.\, q_j^1 \disim q_j^2$.
(This last property uses the fact that $\disim$ propagates forward. 
Direct trace inclusion $\directtraceinclusion$ does not suffice.) 
By induction hypothesis, there exists an infinite initial trace
$\pi^3 = q_0^3 \goesto {\symb_0} q_1^3 \goesto {\symb_1} \cdots$ 
on $w$ that is $i$-good and $\mathcal C^{\mathrm {di}}(\pi^2, \pi^3)$ and
$\forall j \ge i.\, q_j^2 \disim q_j^3$. 
By transitivity we obtain $\mathcal C^{\mathrm {di}}(\pi^1, \pi^3)$,
$q_{i+1}^1 \strictdisim q_{i+1}^3$ and
$\forall j \ge i+1.\, q_j^1 \disim q_j^3$.

If $\pi^3$ is $(i+1)$-good then we are done. Otherwise, the argument of
the above paragraph repeats and we obtain an infinite
initial trace $\pi^5 = q_0^5 \goesto {\symb_0} q_1^5 \goesto {\symb_1} \cdots$ 
on $w$ that is $i$-good and $\mathcal C^{\mathrm {di}}(\pi^3, \pi^5)$,
that $q_{i+1}^3 \strictdisim q_{i+1}^5$ and
$\forall j \ge i+1.\, q_j^3 \disim q_j^5$.
This process cannot repeat infinitely often, because this would imply an
infinite strictly increasing $\strictdisim$-chain $q_{i+1}^{2x+1}$ for
$x=0,1,2,\dots$, which is impossible in finite automata.
Therefore, for some finite index $x$, we obtain an infinite
initial trace $\pi^x = q_0^x \goesto {\symb_0} q_1^x \goesto {\symb_1} \cdots$ 
on $w$ that is $(i+1)$-good and, by transitivity, 
$\mathcal C^{\mathrm {di}}(\pi, \pi^x)$
and $\forall j \ge i+1.\, q_j \disim q_j^x$.
Thus $\pi^x$ is the trace $\pi'$ that we were looking for.

Given the infinite fair initial trace $\hat{\pi}$ on $w$ in $\A$,
it follows from property (P) and K\"onig's Lemma that there
exists an infinite initial trace $\tilde{\pi}$ on $w$ that is
$i$-good for every $i$ and $\mathcal C^{\mathrm {di}}(\hat{\pi},
\tilde{\pi})$.
Therefore $\tilde{\pi}$ is an infinite fair initial trace 
on $w$ in $\A'$ and thus $w \in \lang{\A'}$.
\end{proof}

Theorem~\ref{lem:bwtrace-fwsim} implies that 
$\makeprunerel{\strictbwsim}{\strictdisim}$ is GFP,
but $\makeprunerel{\strictbwdirecttraceinclusion}{\strictdirecttraceinclusion}$ 
is not; see Figure~\ref{fig:pruning}. 
Moreover, $\makeprunerel{\id}{\strictdesim}$ is not GFP (even if $\A =\A/\!\desim$);
see Figure~\ref{fig:pruning}.

\input{fig_GFP_counterexamples}

The quotienting and transition pruning techniques described above
use intricate combinations of backward and forward simulations 
(and more general trace inclusions). In particular, they subsume
previous attempts to combine backward and forward simulations for automata
minimization by {\em mediated preorder} \cite{AbdullaCHV09} (but not
vice-versa).
Mediated preorder is defined as the largest fragment 
$M \subseteq \disim \circ (\bwsim)^{-1}$ s.t. $M \circ \disim \subseteq M$.
In particular, $M$ is a preorder that is GFQ. However, an automaton $\A$ that has
been minimized by the techniques described above cannot be further reduced by
mediated preorder. First we have $\A = \A/\!\bwsim = \A/\!\disim$ by repeated
quotienting. Second, there cannot exist any distinct states $x,y$ in $\A$
s.t. $(x \disim y \wedge x \bwsim y)$ by 
the pruning techniques above (used with simulations as approximations for
trace inclusions) and the removal of dead states.
Under these conditions, quotienting with mediated preorder has no effect,
as the following theorem shows.

\begin{theorem}
Let $\A$ be an automaton s.t.
(1) $\A = \A/\!\bwsim = \A/\!\disim$ and
(2) $x \disim y \wedge x \bwsim y \Rightarrow x=y$.
Then $\A = \A/\! M$.
\end{theorem}
\begin{proof}
We show that 
$x M y \wedge y M x \Rightarrow x=y$ which implies $\A = \A/\! M$.

Let $x M y$ and $y M x$. By definition of $M$ 
there exist mediators $z$ s.t. $x \disim z$ and $y \bwsim z$,
and $w$ s.t. $x \bwsim w$ and $y \disim w$.
Since $M \circ \disim \subseteq M$ we have $x M w$. Thus there exists a mediator
$k$ s.t. $x \disim k$ and $w \bwsim k$.
By transitivity of $\bwsim$ we also have $x \bwsim k$.
By (2) we get $x=k$.
Thus $x \bwsim w$ and $w \bwsim x$.
By (1) we get $x=w$.
Thus $y \disim w=x \disim z$ and by transitivity $y \disim z$. 
Moreover, $y \bwsim z$ as above.
By (2) we get $z=y$.
Thus $x \disim z = y$ and $y \disim w = x$.
By (1) we get $x=y$.
\end{proof}

%% file: GFP_table.tex
\begin{figure}
	
	\begin{center}
		\begin{tabular}{c|ccc||cc}
			$\brel\backslash\frel$
								&	$\id$		&	$\strictdisim$	&	$\strictdirecttraceinclusion$
																					&	$\strictdesim$	&	$\strictfsim$	\\
			\hline
			$\id$				&	$\tickNO$	&	$\tickOK$		&	$\tickOK$	&	$\tickNO$		&	$\tickNO$		\\
			$\strictbwsim$		&	$\tickOK$	&	$\tickOK$		&	$\tickOK$	&	$\tickNO$		&	$\tickNO$		\\
			$\strictbwdirecttraceinclusion$
								&	$\tickOK$	&	$\tickOK$		&	$\tickNO$	&	$\tickNO$		&	$\tickNO$
		\end{tabular}
	\end{center}
	%
	%
	
	\caption{GFP relations $\makeprunerel{\brel}{\frel}$}
	\label{fig:GFP_relations}

\end{figure}

%% file: fig_GFP_counterexamples.tex
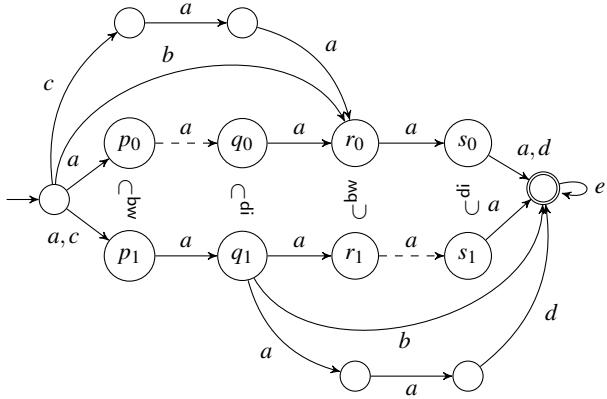
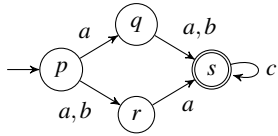
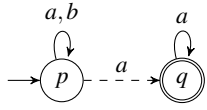
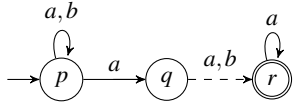
\begin{figure}[t]


\subfigure[
	$\makeprunerel{\strictbwdirecttraceinclusion}{\strictdirecttraceinclusion}$ is not GFP:
	If the dashed transitions $p_0 \goesto a q_0$ and $r_1 \goesto a s_1$ are removed,
	then $a^5e^\omega$ is no longer accepted.
	Note that $\A = \A/\!\bwdirecttraceinclusion = \A/\!\directtraceinclusion$.
	(This example even holds for $\stricttranskbwsim, \stricttranskdisim$
        and $k=3$; cf. Section~\ref{sec:lookahead}).
]{
	\begin{tikzpicture}[on grid, node distance=1cm and 1.5cm]
		\tikzstyle{vertex} = [smallstate]

		\path node [vertex, initial] (i) {};
		
		\path node [vertex] (p0) [above right = .75 and 1 cm of i] {$p_0$};
		\path node [vertex] (q0) [right = of p0] {$q_0$};
		\path node [vertex] (r0) [right = of q0] {$r_0$};
		\path node [vertex] (s0) [right = of r0] {$s_0$};
		
		\path node [vertex] (p1) [below right = .75 and 1 cm of i] {$p_1$};
		\path node [vertex] (q1) [right = of p1] {$q_1$};
		\path node [vertex] (r1) [right = of q1] {$r_1$};
		\path node [vertex] (s1) [right = of r1] {$s_1$};
		
		\path node [vertex, accepting] (f) [below right = .6 and 1 cm of s0] {};
		
		\path node [vertex] (x0) [above = 1.6cm of p0] {};
		\path node [vertex] (y0) [right = of x0] {};

		\path node [vertex] (x1) [below = 1.6cm of r1] {};
		\path node [vertex] (y1) [right = of x1] {};

		\path[->]

			(i) edge node [above left] {$a$} (p0)
			(i) edge [bend left = 30] node [above left] {$c$} (x0)
			(i) edge [bend left = 75] node [above left] {$b$} (r0)
			
			(p0) edge [dashed] node [above] {$a$} (q0)
			(q0) edge node [above] {$a$} (r0)
			(r0) edge node [above] {$a$} (s0)
			(s0) edge node [above right] {$a,d$} (f)
			
			(x0) edge node [above] {$a$} (y0)
			(y0) edge [bend left = 30] node [above right] {$a$} (r0)
			
			(i) edge node [below left] {$a,c$} (p1)
			(q1) edge [bend right = 30] node [below left] {$a$} (x1)
			(q1) edge [bend right = 75] node [below left] {$b$} (f)
			
			(p1) edge node [above] {$a$} (q1)
			(q1) edge node [above] {$a$} (r1)
			(r1) edge [dashed] node [above] {$a$} (s1)
			(s1) edge node [above left] {$a$} (f)
			
			(x1) edge node [below] {$a$} (y1)
			(y1) edge [bend right = 30] node [below right] {$d$} (f)
			
			(f) edge [loop right] node {$e$} ();

		\path
			(p0) -- node [midway] {\rotatebox{-90}{$\strictbwdirecttraceinclusion$}} (p1)
			(q0) -- node [midway] {\rotatebox{-90}{$\strictdirecttraceinclusion$}} (q1)
			(r0) -- node [midway] {\rotatebox{90}{$\strictbwdirecttraceinclusion$}} (r1)
			(s0) -- node [midway] {\rotatebox{90}{$\strictdirecttraceinclusion$}} (s1);

		\begin{pgfonlayer}{background}
		\end{pgfonlayer}

	\end{tikzpicture}
}
\\
\subfigure[
	GFP is not closed under union: Pruning automaton $\A$ with 
	$\makeprunerel{\id}{\strictdisim} \cup \makeprunerel{\strictbwsim}{\id}$
	would remove the transitions $p \goesto a r$ and $q \goesto a s$,
	and thus $aac^\omega$ would no longer be accepted.
]{
	\begin{tikzpicture}[on grid, node distance= .6cm and 1cm]
		\tikzstyle{vertex} = [smallstate]

		\path node [vertex, initial] (p) {$p$};
		\path node [vertex] (q) [above right = of p] {$q$};
		\path node [vertex] (r) [below right = of p] {$r$};
		\path node [vertex, accepting] (s) [below right = of q] {$s$};

		\path[->]

			(p) edge node [above left] {$a$} (q)
			(q) edge node [above right] {$a,b$} (s)
			(p) edge node [below left] {$a,b$} (r)
			(r) edge node [below right] {$a$} (s)

			(s) edge [loop right] node {$c$} ();

	\end{tikzpicture}
}
$\ $
\subfigure[
	$\makeprunerel{\id}{\strictdesim}$ is not GFP:
	We have $q \strictdesim p$, but removing the dashed transition $p \goesto a q$ makes the language empty,
	even though ${\cal A} = {\cal A}/\!\desim$.
]{	$\ \ $
	\begin{tikzpicture}[on grid, node distance= .6cm and 1.6cm]
		\tikzstyle{vertex} = [smallstate]

		\path node [vertex, initial] (p) {$p$};
		\path node (hidden) (x) [below = of p] {}; 
		\path node [vertex, accepting] (q) [right = of p] {$q$};

		\path[->]

			(p) edge [dashed] node [above] {$a$} (q)
			(p) edge [loop above] node {$a,b$} ()
			(q) edge [loop above] node {$a$} ();

	\end{tikzpicture}
	$\ \ $
}
\\
\centering
\subfigure[
	$\makeprunerel{\strictbwsim}{\strictlanguageinclusion}$ is not GFP:
	In the automaton above, both transitions $p \goesto a q$ and $q \goesto a r$ are transient.
	Moreover, $r \strictlanguageinclusion q$ (even $r \strictdesim q$) and $q \strictbwsim p$.
	However, removing the smaller transition $q \goesto a r$ changes the language,
	since $a^\omega$ is no longer accepted.
	Thus, $\makeprunerel{\strictbwsim}{\strictlanguageinclusion}$ is not GFP
	even when one restricts to comparing/pruning only transient transitions
	(unlike	$\makeprunerel{\id}{\strictlanguageinclusion}$).
]{	$\qquad\qquad\quad$
	\begin{tikzpicture}[on grid, node distance= .6cm and 1.4cm]
		\tikzstyle{vertex} = [smallstate]

		\path node [vertex, initial] (p) {$p$};
		\path node [vertex, initial] (q) [right = of p] {$q$};
		\path node [vertex, accepting] (r) [right = of q] {$r$};

		\path[->]

			(p) edge node [above] {$a$} (q)
			(q) edge [dashed] node [above] {$a,b$} (r)
			(p) edge [loop above] node {$a,b$} ()
			(r) edge [loop above] node {$a$} ();

	\end{tikzpicture}
	$\qquad\qquad\quad$
}

	\caption{Pruning counterexamples.}
	\label{fig:pruning}

\end{figure}

%% file: lookahead.tex
\section{Lookahead Simulations}\label{sec:lookahead}

While trace inclusions are theoretically appealing as GFQ/GFI preorders coarser than simulations,
it is not feasible to use them in practice, because they are too hard to
compute (even their membership problem is PSPACE-complete).
As a first attempt at achieving a better trade-off between complexity and size
we recall \emph{multipebble simulations} \cite{etessami:hierarchy02},
which are obtained by providing Duplicator with several pebbles, instead of one.
However, computing multipebble simulations is not feasible in practice either,
on automata of nontrivial size.
Therefore, we explore yet another way of obtaining good under-approximations of trace inclusion:
We introduce \emph{lookahead simulations},
which are obtained by providing Duplicator with a limited amount of information about Spoiler's future moves.
While lookahead itself is a classic concept (e.g., in parsing) it can be
defined in several different ways in the context of adversarial games like in
simulation. We compare different variants for computational efficiency and approximation quality.

\xparagraph{$k$-pebble simulation.}

Simulation preorder can be generalized by allowing Duplicator to control several pebbles instead of just one. 
In $k$-pebble simulation, $k>0$, Duplicator's position is a set of at most $k$ states
(while Spoiler still controls exactly 1 state),
which allows Duplicator to `hedge her bets' in the simulation game.
The direct, delayed, fair and backward winning conditions can be generalized to the multipebble framework \cite{etessami:hierarchy02}.
For $x\in\{\mathrm{di, de, f, bw}\}$ and $k > 0$,
$k$-pebble $x$-simulation is coarser than $x$-simulation and it implies $x$-containment;
by increasing $k$, one can control the quality of the approximation to trace inclusion.
Direct, delayed, fair and backward $k$-pebble simulations are not transitive in general,
but their transitive closures are GFI preorders;
the direct, delayed and backward variants are also GFQ.
%
However, computing $k$-pebble simulations is infeasible, even for modest values for $k$.
In fact, for a BA with $n$ states, computing $k$-pebble simulation requires solving a game of size $n \cdot n^k$.
Even in the simplest case of $k=2$ this means at least cubic space,
which is not practical for large $n$.
For this reason, we consider a different way to extend Duplicator's power, 
i.e., by using \emph{lookahead} on the moves of Spoiler.

\xparagraph{$k$-step simulation.}

We generalize simulation by having the players select sequences of transitions of length $k > 0$ instead of single transitions:
This gives Duplicator more information,
and thus yields a larger simulation relation.
%
%
In general, $k$-step simulation and $k$-pebble simulation are incomparable,
but $k$-step simulation is strictly contained in $n$-pebble simulation.
However, the rigid use of lookahead in big-steps causes at least two issues:
\begin{inparaenum}[1)]
	\item For a BA with $n$ states, we need to store only $n^2$ configurations $(p,q)$
	(which is much less than $k$-pebble simulation).
        However, in {\em every round} we have to explore up-to $d^k$ different moves for each player
	(where $d$ is the maximal out-degree of the automaton). 
        In practice (e.g., $d=4$, $k=12$) this is still too large.
	\item Duplicator's lookahead varies between $1$ and $k$,
	depending where she is in her response to Spoiler's long move.
	Thus, Duplicator might lack lookahead where it is most needed,
	while having a large lookahead in other situations where it is not useful.
	In the next notion, we attempt at ameliorating this.
\end{inparaenum}

\xparagraph{$k$-continuous simulation.}

Duplicator is continuously kept informed about Spoiler's next $k$ moves,
i.e., she always has lookahead $k$.
Formally, a configuration of the simulation game consists in a pair $(\rho_i, q_i)$,
where $\rho_i$ is the sequence of the next $k-1$ moves from $p_i$ that Spoiler has already committed to.
In every round of the game, Spoiler reveals another move $k$ steps in the future,
and then makes the first of her announced $k$ moves,
to which Duplicator responds as usual.
A pair of states $(p,q)$ is in $k$-continuous simulation
if Duplicator can win this game from every configuration $(\rho,q)$,
where $\rho$ is a sequence of $k-1$ moves from $p$.
($k=1$ is ordinary simulation.)
$k$-continuous simulation is strictly contained in $n$-pebble simulation
(but incomparable with $k$-pebble simulation),
and larger than $k$-step simulation.
While this is arguably the strongest way of giving lookahead to Duplicator,
it requires storing $n^2 \cdot d^{k-1}$ configurations, which is infeasible
for nontrivial $n$ and $k$ (e.g., $n=10000$, $d=4$, $k=12$).

\xparagraph{$k$-lookahead simulation.}

We introduce $k$-lookahead simulation as an optimal compromise between $k$-step and $k$-continuous simulation.
Intuitively, we put the lookahead under Duplicator's control,
who can choose at each round how much lookahead she needs (up to $k$).
Formally, configurations are pairs $(p_i, q_i)$ of states.
In every round of the game, Spoiler chooses a sequence of $k$ consecutive transitions
$p_i \goesto {\symb_i} {p_{i+1}} \goesto {\symb_{i+1}} \cdots \goesto {\symb_{i+k-1}} p_{i+k}$.
Duplicator then chooses a number $1 \leq m \leq k$ and responds with a matching sequence of $m$ transitions 
$q_i \goesto {\symb_i} {q_{i+1}} \goesto {\symb_{i+1}} \cdots \goesto {\symb_{i+m-1}} q_{i+m}$.
The remaining $k-m$ moves of Spoiler are forgotten,
and the next round of the game starts at $(p_{i+m}, q_{i+m})$.
In this way, the players build two infinite traces $\pi_0$ from $p_0$ and $\pi_1$ from $q_0$.
Backward simulation is defined similarly with backward transitions.
For acceptance condition $x \in \{\mathrm{di, de, f, bw}\}$,
Duplicator wins this play if $\mathcal C^x(\pi_0, \pi_1)$ holds.
\begin{definition}\label{def:lookahead-sim}
	Two states $(p_0, q_0)$ are in \emph{$k$-lookahead $x$-simulation},
	written $p_0 \kxsim k x q_0$,
	iff	Duplicator has a winning strategy in the above game.
\end{definition}
Since $\kxsim k x$ is not transitive (unless $k=1$; cf. Appendix~\ref{sec:non_transitivity}),
we denote its transitive closure, which is a preorder, by $\transksimx$,
and its asymmetric restriction by $\stricttransksimx = \transksimx \setminus (\transksimx)^{-1}$.

Lookahead simulation offers the optimal trade-off between $k$-step and $k$-continuous simulation.
Since the lookahead is discarded at each round,
$k$-lookahead simulation is (strictly) included in $k$-continuous lookahead
(where the lookahead is never discarded).
However, this has the benefit of only requiring to store $n^2$ configurations,
which makes computing lookahead simulation space-efficient.
On the other side, when Duplicator always chooses a maximal reply $m = k$
we recover $k$-step simulation,
which is thus included in $k$-lookahead simulation.
Moreover, thanks to the fact that Duplicator controls the lookahead,
most rounds of the game can be solved without ever reaching the maximal lookahead $k$:
\begin{inparaenum}[1)]
	\item for a fixed attack by Spoiler,
	we only consider Duplicator's responses for small $m = 1, 2, \dots, k$
	until we find a winning one, and
	\item also Spoiler's attacks can be built incrementally	since,
	if she loses for some lookahead $h$, then she also loses for $h'\!\! >\!\! h$.
\end{inparaenum}
In practice, this greatly speeds up the computation,
and allows us to use lookaheads in the range $4$-$25$,
depending on the size and structure of the automata;
see Section~\ref{sec:experiments} for the experimental evaluation and
benchmark against the GOAL tool \cite{GOAL_survey_paper}.

$k$-lookahead simulation can also be expressed as a restriction of $n$-pebble simulation,
where Duplicator is allowed to split pebbles maximally (thus $n$-pebbles),
but after a number $m \le k$ rounds (where $m$ is chosen dynamically by Duplicator)
she has to discard all but one pebble.
Then, Duplicator is allowed to split pebbles maximally again, etc.
Thus, $k$-lookahead simulation is contained in $n$-pebble simulation,
though it is generally incomparable with $k$-pebble simulation.

Direct, delayed, fair and backward $k$-lookahead simulation have 
a fixed-point characterization expressible in $\mu$-calculus (cf. Appendix~\ref{sec:fixedpoint}),
which can be useful for a symbolic implementation. However, our current 
algorithm computes them with an explicit-state representation.

%% file: heavyandlight.tex
\section{Automata Minimization} \label{sec:heavyandlight}

\ignore{

In addition to the transition pruning techniques of
Section~\ref{sec:minimization},
$k$-lookahead simulations can also be used for quotienting.

By \cite{etessami:hierarchy02}, $\directtraceinclusion$ and $n$-pebble
delayed simulation are GFQ, and, by Theorem~\ref{lem:bwincl_GFQ}, 
so is $\bwdirecttraceinclusion$.
While these relations are PSPACE-complete, they can be
efficiently underapproximated by the PTIME-computable 
lookahead simulations $\transkdisim$,
$\transkdesim$ and $\transkbwsim$, respectively.
In particular, $\transkdisim$ and $\transkdesim$ and $\transkbwsim$ are GFQ.

}

We minimize automata by transition pruning and quotienting.
While trace inclusions would be an ideal basis for such techniques,
they (i.e., their membership problems) are PSPACE-complete.
Instead, we use lookahead simulations as efficiently computable under-approximations;
in particular, we use
\begin{itemize}
	\item $\transkdisim$ in place of direct trace inclusion $\directtraceinclusion$ (which is GFQ \cite{etessami:hierarchy02}).

	\item $\transkdesim$ in place of $n$-pebble delayed simulation (GFQ \cite{etessami:hierarchy02}).

	\item $\transkfsim$ in place of fair trace inclusion $\fairtraceinclusion$ (which is GFI).

	\item $\transkbwsim$ in place of backward trace inclusion $\bwdirecttraceinclusion$
	(which is GFQ by Theorem~\ref{lem:bwincl_GFQ}).
\end{itemize}
For pruning, we apply the results of Section~\ref{sec:minimization} and the substitutions above
to obtain the following GFP relations: 
\begin{align*}
	\makeprunerel{\id}{\stricttranskdisim},
	\makeprunerel{\stricttranskbwsim}{\id},
	\makeprunerel{\strictbwsim}{\transkdisim},
	\makeprunerel{\transkbwsim}{\strictdisim},
	R_t(\stricttranskfsim)
\end{align*}
For quotienting, we employ delayed $\transkdesim$ and backward $\transkbwsim$ $k$-lookahead simulations (which are GFQ).
Below, we describe two possible ways to combine our simplification techniques: \emph{Heavy-$k$} and \emph{Light-$k$}
(which are parameterized by the lookahead value $k$).

\paragraph{Heavy-$k$.}

We advocate the following minimization procedure,
which repeatedly applies all the techniques described in this paper
until a fixpoint is reached:
\begin{inparaenum}[1)]
	\item Remove dead states.
	\item Prune transitions w.r.t. the GFP relations above (using lookahead $k$).
	\item Quotient w.r.t. $\transkdesim$ and $\transkbwsim$.
\end{inparaenum}
The resulting simplified automaton cannot be further reduced by any of these techniques.
In this sense, it is a local minimum in the space of automata.
Applying the techniques in a different order might produce a different local minimum,
and, in general, there does not exist an optimal order that works best in every instance.
In practice, the order is determined by efficiency considerations and easily computable operations are used first \cite{ourtest,RABIT}.

\begin{remark}
	While quotienting with ordinary simulation is idempotent, in general this is not true for lookahead simulations,
	because these relations are not preserved under quotienting (unlike ordinary simulation).
	Moreover, quotienting w.r.t. forward simulations does not preserve backward simulations, and vice-versa.
	Our experiments showed that repeatedly and alternatingly quotienting w.r.t. $\transkdesim$ and $\transkbwsim$
	(in addition to our pruning techniques) yields the best minimization effect.
\end{remark}

The Heavy-$k$ procedure {\em strictly subsumes} all simulation-based automata
minimization methods described in the literature 
(removing dead states, quotienting, pruning of `little brother' transitions, 
mediated preorder), except for the following two:
\begin{inparaenum}[1)]
	\item The \emph{fair simulation minimization} of \cite{GBS02} works by tentatively 
	merging fair simulation equivalent states and then checking if this operation 
	preserved the language. (In general, fair simulation is not
	GFQ.) It subsumes quotienting with $\desim$ (but not $\transkdesim$) and 
	is implemented in GOAL \cite{GOAL_survey_paper}.
	We benchmarked our methods against it and found Heavy-$k$ to be much better in
	both effect and efficiency; cf Section~\ref{sec:experiments}.
	\item The GFQ \emph{jumping-safe preorders} of \cite{buchiquotient:ICALP11,Clemente:PhD} are incomparable to
	the techniques described in this paper. If applied in addition to Heavy-$k$,
	they yield a very modest extra minimization effect.
\end{inparaenum}

\paragraph{Light-$k$.}

This procedure is defined purely for comparison reasons.
It demonstrates the effect of the lookahead $k$
in a single quotienting operation and works as follows:
Remove all dead states and then quotient w.r.t. $\transkdesim$.
Although Light-$k$ achieves much less than Heavy-$k$, it is not necessarily faster.
This is because it uses the more expensive to compute relation $\transkdesim$
directly, while Heavy-$k$ applies other cheaper (pruning) operations
first and only then computes $\transkdesim$ on the resulting smaller
automaton.

%% file: inclusion.tex
\section{Language Inclusion Checking}\label{sec:preprocessing}

The language inclusion problem $\A \languageinclusion \B$
is PSPACE-complete \cite{kupfermanvardi:fair_verification}.
It can be solved via complementation of $\B$ \cite{sistla:vardi:wolper:complementation:87,GOAL_survey_paper}
and, more efficiently, by rank-based (\cite{fogarty_et_al:LIPIcs:2011:3235} and references therein)
or Ramsey-based methods
\cite{seth:buchi,seth:efficient,abdulla:simulationsubsumption,Rabit_CONCUR2011},
or variants of Piterman's construction \cite{Pit06,GOAL_survey_paper}.
Since these all have {\em exponential} time complexity, it helps significantly
to first minimize the automata in a preprocessing step.
Better minimization techniques, as described in the previous sections, make it
possible to solve significantly larger instances.
However, our simulation-based techniques can not only be used in
preprocessing, but actually solve most instances of the inclusion problem
{\em directly}. This is significant, because simulation scales 
{\em polynomially} (quadratic average-case complexity; cf. Section~\ref{sec:experiments}).

\input{preprocessing}

\input{jumpsim}

\input{inclalg}

%% file: preprocessing.tex
\subsection{Inclusion-preserving minimization}

Inclusion checking algorithms generally benefit from language-preserving minimization preprocessing
(cf. Sec.~\ref{sec:heavyandlight}).
However, preserving the languages of $\A$ and $\B$ in the preprocessing is not actually necessary.
A preprocessing on $\A,\B$ is said to be \emph{inclusion-preserving}
iff it produces automata $\A',\B'$ s.t. $\A \languageinclusion \B \iff \A' \languageinclusion \B'$
(regardless of whether $\A \languageequivalence \A'$ or $\B \languageequivalence \B'$).
In the following, we consider two inclusion-preserving preprocessing steps.

\xparagraph{Simplify $\A$.}

In theory, the problem $\A \languageinclusion \B$ is only hard in $\B$,
but polynomial in the size of $\A$. However, this is only relevant if one actually 
constructs the exponential-size complement of $\B$, which is of course to be
avoided. For polynomial simulation-based algorithms it is crucial to also minimize $\A$.
The idea is to remove transitions in $\A$ which are `covered' by better transitions in $\B$.

\begin{definition}
	Given $\A = (\Sigma, Q_\A, I_\A, F_\A, \delta_\A)$, 
	$\B = (\Sigma, Q_\B, I_\B, F_\B, \delta_\B)$,
	let $\prunerel \subseteq \delta_\A \times \delta_\B$ be a relation for comparing transitions in $\A$ and $\B$.
	The pruned version of $\A$ is $\xprune{\A}{\B}{\prunerel} := (\Sigma, Q_\A, I_\A, F_\A, \delta')$
	with $\delta' = \{(p,\symb,r) \in \delta_\A \st {\nexists}
	(p',\symb',r') \in \delta_\B.\, (p,\symb,r) \prunerel (p',\symb',r')\}$.
\end{definition}

\noindent
$\A \languageinclusion \B$ implies $\xprune{\A}{\B}{\prunerel} \languageinclusion \B$
(since $\xprune{\A}{\B}{\prunerel} \languageinclusion \A$).
When also the other direction holds (so pruning is inclusion-preserving),
we say that $\prunerel$ is \emph{good for $\A,\B$-pruning},
i.e., when $\A \languageinclusion \B \iff \xprune{\A}{\B}{\prunerel} \languageinclusion \B$.
Intuitively, pruning is correct when the removed edges do not allow $\A$ to accept any word
which is not already accepted by $\B$.
In other words, if there is a counter example to inclusion in $\A$,
then it can even be found in $\xprune{\A}{\B}{\prunerel}$.
As in Sec.~\ref{sec:minimization}, we compare transitions by looking at their endpoints:
For state relations $\brel, \frel \subseteq Q_\A \times Q_\B$, let
$\makeprunerel{\brel}{\frel} = \{((p,\symb,r),(p',\symb,r')) \st p \brel p'\,\wedge\, r \frel r'\}$.

Since inclusion-preserving pruning does not have to respect the language,
we can use much weaker (i.e., coarser) relations for comparing endpoints.
Let $\accblindbwdirecttraceinclusion$ be the variant of $\bwdirecttraceinclusion$
where accepting states are not taken into consideration.
\begin{theorem}\label{lem:ABpruning}
$\makeprunerel{\accblindbwdirecttraceinclusion}{\fairtraceinclusion}$ is good for $\A,\B$-pruning.
\end{theorem}
\begin{proof}
	Let $\prunerel = \makeprunerel{\accblindbwdirecttraceinclusion}{\fairtraceinclusion}$.
	One direction is trivial.
	For the other direction, by contraposition, assume $\xprune{\A}{\B}{\prunerel} \languageinclusion \B$, but
	$\A \not\languageinclusion \B$. There exists a $w \in \lang{\A}$
	s.t. $w \notin \lang{\B}$. There exists an initial fair trace 
	$\pi = q_0 \goesto {\symb_0} q_1 \goesto {\symb_1} \cdots$ on $w$ in $\A$.
	There are two cases.
	\begin{enumerate}
		\item
		$\pi$ does not contain any transition
		$q_i \goesto {\symb_i} q_{i+1}$ that is not present in
		$\xprune{\A}{\B}{\prunerel}$. Then $\pi$ is also an initial fair trace
		on $w$ in $\xprune{\A}{\B}{\prunerel}$, and thus we obtain
		$w \in \lang{\xprune{\A}{\B}{\prunerel}}$ and
		$w \in \lang{\B}$. Contradiction.
		\item
		$\pi$ contains a transition
		$q_i \goesto {\symb_i} q_{i+1}$ that is not present in
		$\xprune{\A}{\B}{\prunerel}$.
		Therefore there exists a transition
		$q_i' \goesto {\symb_i} q_{i+1}'$ in $\B$ s.t.
		$q_i \accblindbwdirecttraceinclusion q_i'$ and $q_{i+1} \fairtraceinclusion q_{i+1}'$.
		Thus there exists an initial fair trace on $w$ in $\B$ and thus
		$w \in \lang{\B}$. Contradiction. \qedhere
	\end{enumerate}
\end{proof}

\noindent
We can approximate $\accblindbwdirecttraceinclusion$ with
(the transitive closure of) a corresponding $k$-lookahead simulation $\accblindkbwsim$,
which is defined as $\kbwsim$, except that only initial states are considered,
i.e., the winning condition is
$\mathcal C^\mathrm{bw-}(\pi_0, \pi_1) \iff \forall (i \geq 0) \cdot p_i \in I \implies q_i \in I$.
Let $\accblindtranskbwsim$ be the transitive closure of $\accblindkbwsim$.
Since GFP is $\subseteq$-downward closed and $\makeprunerel{\cdot}{\cdot}$ is monotone,
we get this corollary.

\begin{corollary}
	$\makeprunerel{\accblindtranskbwsim}{\transkfsim}$ is good for $\A,\B$-pruning.
\end{corollary}

\xparagraph{Simplify $\B$.}

Let $\A \times \B$ be the synchronized product of $\A$ and $\B$.
The idea is to remove states in $\B$ which cannot be reached in $\A \times \B$.
Let $R$ be the set of states in $\A \times \B$ reachable from $I_\A \times I_\B$,
and let $X \subseteq Q_\B$ be the projection of $R$ to the $\B$-component.
We obtain $\B'$ from $\B$ by removing all states $\notin X$ and their associated transitions.
Although $\B' \not\languageequivalence \B$, this operation is clearly inclusion-preserving.


%% file: jumpsim.tex
\subsection{Jumping fair simulation as a better GFI relation} \label{sec:jumpsim}

We further generalize the GFI preorder $\transkfsim$ by allowing Duplicator even more freedom.
The idea is to allow Duplicator to take \emph{jumps} during the simulation game (in the spirit of \cite{Clemente:PhD}).
For a preorder $\le$ on $Q$, in the game for \emph{$\le$-jumping $k$-lookahead simulation}
Duplicator is allowed to jump to $\le$-larger states before taking a transition.
Thus, a Duplicator's move is of the form
$q_i \le q_i' \goesto {\symb_i} {q_{i+1}} \le q_{i+1}' \goesto {\symb_{i+1}} \cdots \goesto {\symb_{i+m-1}} q_{i+m}$,
and she eventually builds an infinite $\le$-jumping trace. We say that this trace
is \emph{accepting} at step $i$ iff $\exists q_i'' \in F.\, q_i \le q_i'' \le q_i'$,
and \emph{fair} iff it is accepting infinitely often.
As usual, \emph{$\le$-jumping $k$-lookahead fair simulation} holds
iff Duplicator wins the corresponding game, with the fair winning condition
lifted to jumping traces.

Not all preorders $\le$ induce GFI jumping simulations.
The preorder $\le$ is called {\em jumping-safe} \cite{Clemente:PhD} if,
for every word $w$, there exists a $\le$-jumping initial fair trace on $w$ 
iff there exists an initial fair non-jumping one.
Thus, jumping-safe preorders allows to convert jumping traces into non-jumping ones.
Consequently, for a jumping-safe preorder $\le$,
$\le$-jumping $k$-lookahead fair simulation is GFI.

One can prove that $\bwdirecttraceinclusion$ is jumping-safe, while $\accblindbwdirecttraceinclusion$ is not. 
%
We even improve $\bwdirecttraceinclusion$ to a slightly more general jumping-safe relation $\countingbwtraceinclusion$,
by only requiring that Duplicator visits at least as many accepting states as Spoiler does,
but not necessarily at the same time.
Formally, $p_m \countingbwtraceinclusion q_m$ iff,
for every initial $w$-trace
$\pi_0 = p_0 \goesto {\symb_0} p_1 \goesto {\symb_1} \cdots \goesto {\symb_{m-1}} p_m$, 
there exists an initial $w$-trace
$\pi_1 = q_0 \goesto {\symb_0} q_1 \goesto {\symb_1} \cdots \goesto {\symb_{m-1}} q_m$, 
s.t. $|\{i \,|\, p_i \in F\}| \le |\{i \,|\, q_i \in F\}|$.

\begin{theorem}\label{lem:jumping-fairsim}
	The preorder $\countingbwtraceinclusion$ is jumping-safe.
\end{theorem}
\begin{proof}
	%
	Since $\countingbwtraceinclusion$ is reflexive, the existence of an initial
	fair trace on $w$ directly implies the existence of a 
	$\countingbwtraceinclusion$-jumping initial fair trace on $w$.
	
	Now, we show the reverse implication.
	Given two initial $\countingbwtraceinclusion$-jumping traces on $w$
	$\pi_0 = p_0 \countingbwtraceinclusion p_0' \goesto {\symb_0} {p_{1}} 
	\countingbwtraceinclusion p_{1}' \goesto {\symb_{1}} \cdots$
	and 
	$\pi_1 = q_0 \countingbwtraceinclusion q_0' \goesto {\symb_0} {q_{1}} 
	\countingbwtraceinclusion q_{1}' \goesto {\symb_{1}} \cdots$
	we define $\mathcal C^c_j(\pi_0, \pi_1)$ iff
	$|\{i \le j\,|\, \exists p_i'' \in F.\,  p_i \countingbwtraceinclusion
	p_i'' \countingbwtraceinclusion p_i'\}| \le 
	|\{i \le j\,|\, \exists q_i'' \in F.\,  q_i \countingbwtraceinclusion
	q_i'' \countingbwtraceinclusion q_i'\}|$.
	We say that an initial $\countingbwtraceinclusion$-jumping trace on $w$
	is {\em $i$-good} iff it does not jump within the first $i$ steps.

	We show, by induction on $i$, the following property (P):
	For every $i$ and every 
	infinite $\countingbwtraceinclusion$-jumping initial trace
	$\pi = p_0 \countingbwtraceinclusion p_0' \goesto {\symb_0} {p_{1}} \countingbwtraceinclusion p_{1}' \goesto {\symb_{1}} \cdots$
	on $w$ there exists 
	an initial $i$-good trace $\pi^i = q_0 \goesto {\symb_0} q_1 \goesto
	{\symb_1} \cdots \goesto {\symb_i} q_i \cdots$ on $w$
	s.t. $\mathcal C^c_i(\pi, \pi^i)$ and the suffixes of the traces are identical, i.e.,
	$q_i = p_i$ and $\suffix \pi i = \suffix {\pi^i} i$.

	For the case base $i=0$ we take $\pi^0 = \pi$. 
	Now we consider the induction step. 
	By induction hypothesis we get an initial $i$-good trace $\pi^i$
	s.t. $\mathcal C^c_i(\pi, \pi^i)$ and $q_i = p_i$ and 
        $\suffix \pi i = \suffix {\pi^i} i$.
	If $\pi^i$ is $(i+1)$-good then we can take $\pi^{i+1} = \pi^{i}$.
	Otherwise, $\pi^i$ contains a step 
	$q_i \countingbwtraceinclusion q_i' \goesto {\symb_i} {q_{i+1}}$.
	First we consider the case where there exists a $q_i'' \in F$
	s.t. $q_i \countingbwtraceinclusion q_i'' \countingbwtraceinclusion q_i'$.
	(Note that the $i$-th step in $\pi^i$ can count as accepting in $\mathcal C^c$
	because $q_i'' \in F$, even if $q_i$ and $q_i'$ are not accepting.)
	By def. of $\countingbwtraceinclusion$ there exists
	an initial trace $\pi''$ on a prefix of $w$ that ends in $q_i''$
	and visits accepting states at least as often as the non-jumping
	prefix of $\pi^i$ that ends in $q_i$.
	Again by definition of $\countingbwtraceinclusion$ there exists
	an initial trace $\pi'$ on a prefix of $w$ that ends in $q_i'$
	and visits accepting states at least as often as $\pi''$.
	Thus $\pi'$ visits accepting states at least as often as the {\em jumping}
	prefix of $\pi^i$ that ends in $q_i'$ (by the definition of $\mathcal C^c$).
	By composing the traces we get $\pi^{i+1} = \pi' (q_i' \goesto {\symb_i}
	{q_{i+1}}) \suffix {\pi^i} {i+1}$. Thus $\pi^{i+1}$ is an $(i+1)$-good initial trace
	on $w$ and 
        $\suffix \pi {i+1} = \suffix {\pi^i} {i+1} = \suffix {\pi^{i+1}} {i+1}$ and 
	$\mathcal C^c_{i+1}(\pi^i, \pi^{i+1})$ and $\mathcal C^c_{i+1}(\pi, \pi^{i+1})$.
	The other case where there is no $q_i'' \in F$
	s.t. $q_i \countingbwtraceinclusion q_i'' \countingbwtraceinclusion q_i'$ is
	similar, but simpler.

	Let $\pi$ be an initial $\countingbwtraceinclusion$-jumping fair trace on $w$.
	By property (P) and K\"onig's Lemma there 
	exists an infinite initial non-jumping fair trace $\pi'$ on $w$.
	Thus $\countingbwtraceinclusion$ is jumping-safe.
	\ignore{ 
	Now, we show that $\countingbwtraceinclusion$-jumping $k$-lookahead fair
	simulation is GFI.
	Consider two B\"uchi automata $\A$ and $\B$ s.t. for 
	every $p \in I_\A$ there exists $q \in I_\B$ 
	s.t. $q$ can $\countingbwtraceinclusion$-jumping $k$-lookahead fair simulate
	$p$. If $w \in \lang{\A}$ then there is an infinite fair trace on $w$
	from some $p \in I_\A$. It follows that there is an infinite fair
	$\countingbwtraceinclusion$-jumping trace on $w$ from some $q \in I_\B$.
	Since $\countingbwtraceinclusion$ is jumping-safe, there exists
	an infinite fair initial trace on $w$ in $\B$, i.e., 
	an infinite fair trace on $w$ from some $q' \in I_\B$.
	Thus, $w \in \lang{\B}$ as required.
	}
\end{proof}

\noindent
As a direct consequence, $\countingbwtraceinclusion$-jumping $k$-lookahead fair simulation is GFI.
Since $\countingbwtraceinclusion$ is difficult to compute,
we approximate it by a corresponding lookahead-simulation $\countingkbwsim$ which, in the same spirit,
counts and compares the number of visits to accepting states in every round of the $k$-lookahead backward simulation game.
Let $\countingtranskbwsim$ be the transitive closure of $\countingkbwsim$.

\begin{corollary}
	$\countingtranskbwsim$-jumping $k$-lookahead fair sim. is GFI.
\end{corollary}

%% file: inclalg.tex
\subsection{Advanced inclusion checking algorithm} \label{sec:inclalg}

Given these techniques, we propose the following algorithm for inclusion checking $\A \languageinclusion \B$.

\begin{enumerate}
	
	\item[(1)]
		Use the Heavy-$k$ procedure to minimize $\A$ and $\B$,
		and additionally apply the inclusion-preserving minimization techniques from Sec.~\ref{sec:preprocessing}.
		Lookahead simulations are computed not only on $\A$ and $\B$, but
		also {\em between} them (i.e., on their disjoint union).
		Since they are GFI, we check whether they already witness
                inclusion. Since many simulations are computed between partly
                minimized versions of $\A$ and $\B$, this witnesses
                inclusion much more often than checking fair simulation
                between the original versions. 
		\ignore{
			Note that simulations are computed several times on different partly minimized
			variants of $\A$ and $\B$ during the iterated minimization procedure,
			and simulation is easier to establish on the smaller simpler variants.
			Thus one has more than one chance to establish inclusion, i.e., it is much more
			effective than just checking (fair) simulation between the initial states of the
			original versions of $\A$ and $\B$. Moreover, lookahead-simulations 
			are much larger than ordinary simulations, further increasing the chance to
			establish inclusion already at this stage (of course only in those instances where
			inclusion actually holds); cf Section~\ref{sec:experiments}.
		}
		This step either stops showing inclusion,
		or produces smaller inclusion-equivalent automata $\A', \B'$.
		
	\item[(2)]
		Check the GFI $\countingtranskbwsim$-jumping $k$-lookahead fair simulation from Sec.~\ref{sec:jumpsim} between $\A'$ and $\B'$,
		and stop if the answer is yes.
		
	\item[(3)]
	If inclusion was not established in steps (1) or (2) then try to find a counterexample
	to inclusion. This is best done by a Ramsey-based method (optionally using
	simulation-based subsumption techniques), e.g.,
	\cite{Rabit_CONCUR2011,RABIT}. 
	Use a small timeout value, since in most
	non-included instances there exists a very short counterexample.
	Stop if a counterexample is found.
	
	\item[(4)]
	If steps (1)-(3) failed (rare in practice), use any complete method,
	(e.g., Rank-based, Ramsey-based or Piterman's construction) to check $\A' \languageinclusion \B'$.
	At least, it will benefit from working on the smaller instance $\A', \B'$
	produced by \ignore{the pre-processing} step (1).
	
\end{enumerate}

\noindent
Note that steps (1)-(3) take polynomial time, while step (4) takes exponential time.
(For the latter, we recommend the improved Ramsey method of \cite{Rabit_CONCUR2011,RABIT}
and the on-the-fly variant of Piterman's construction \cite{Pit06} implemented in GOAL \cite{GOAL_survey_paper}.) 
This algorithm allows to solve much larger instances of the
inclusion problem than previous methods
\cite{sistla:vardi:wolper:complementation:87,GOAL_survey_paper,fogarty_et_al:LIPIcs:2011:3235,seth:buchi,seth:efficient,abdulla:simulationsubsumption,Rabit_CONCUR2011,Pit06},
i.e., automata with 1000-20000 states instead of 10-100 states; cf.
Section~\ref{sec:experiments}.

%% file: empirical.tex
\section{Experiments}\label{sec:experiments}

We test the effectiveness of Heavy-k minimization 
on Tabakov-Vardi random automata \cite{tabakov:model}, 
on automata derived from LTL formulae, and on 
automata derived from mutual exclusion protocols,
and compare it to the best previously available 
techniques implemented in GOAL \cite{GOAL_survey_paper}.
A scalability test shows that Heavy-k has quadratic
average-case complexity and it is vastly more efficient than GOAL.
Furthermore, we test our methods for language inclusion
on large instances and compare their performance to
previous techniques. Due to space limitations,
we only give a summary of the results, but all details and
the runnable tools are available \cite{ourtest}.
Unless otherwise stated, the experiments were run with Java 6 on 
Intel Xeon X5550 2.67GHz and 14GB memory.

\paragraph{Random automata.}

The Tabakov-Vardi model \cite{tabakov:model} generates random automata
according to the following parameters:
the number of states $n$, the size of the alphabet $|\Sigma|$, the transition density 
${\it td}$ (number of transitions, relative to $n$ and $|\Sigma|$) and the
acceptance density ${\it ad}$ (percentage of accepting states).
Apart from this, they do not have any special structure, and thus minimization
and language inclusion problem are harder for them than for automata from
other sources (see below).
Random automata provide general reproducible test cases, on average.
Moreover, they are the only test cases that are guaranteed to be unbiased
towards any particular method. Thus, it is a particular sign of quality if a
method performs well even on these hard cases.

The inherent difficulty of the minimization problem,
and thus also the effectiveness of minimization methods, depends strongly on the class of
random automata, i.e., on the parameters listed above. Thus, one needs
to compare the methods over the whole range, not just for one example.
Variations in ${\it ad}$ do not affect Heavy-k much (cf. Appendix~\ref{sec:acceptance_density}),
but very small values make minimization harder for the other methods.
By far the most important parameter 
is ${\it td}$. The following figure shows typical results.
We take $n=100$, $|\Sigma|=2$, ${\it ad}=0.5$ and the range of
${\it td}=1.0, 1.1, \dots, 3.0$.
For each ${\it td}$ we created 300 random automata, minimized them with
different methods, and plotted the resulting average number of states after minimization.
Each curve represents a different method: RD (just remove dead states),
Light-1, Light-12, Heavy-1, and Heavy-12 and GOAL. The GOAL curve
shows the best effort of all previous techniques (as implemented in GOAL),
which include RD, quotienting with backward and forward simulation, pruning of
little brother transitions and the fair simulation minimization of
\cite{GBS02} (which subsumes quotienting with delayed simulation).

\begin{center}
\includegraphics[scale=0.4]{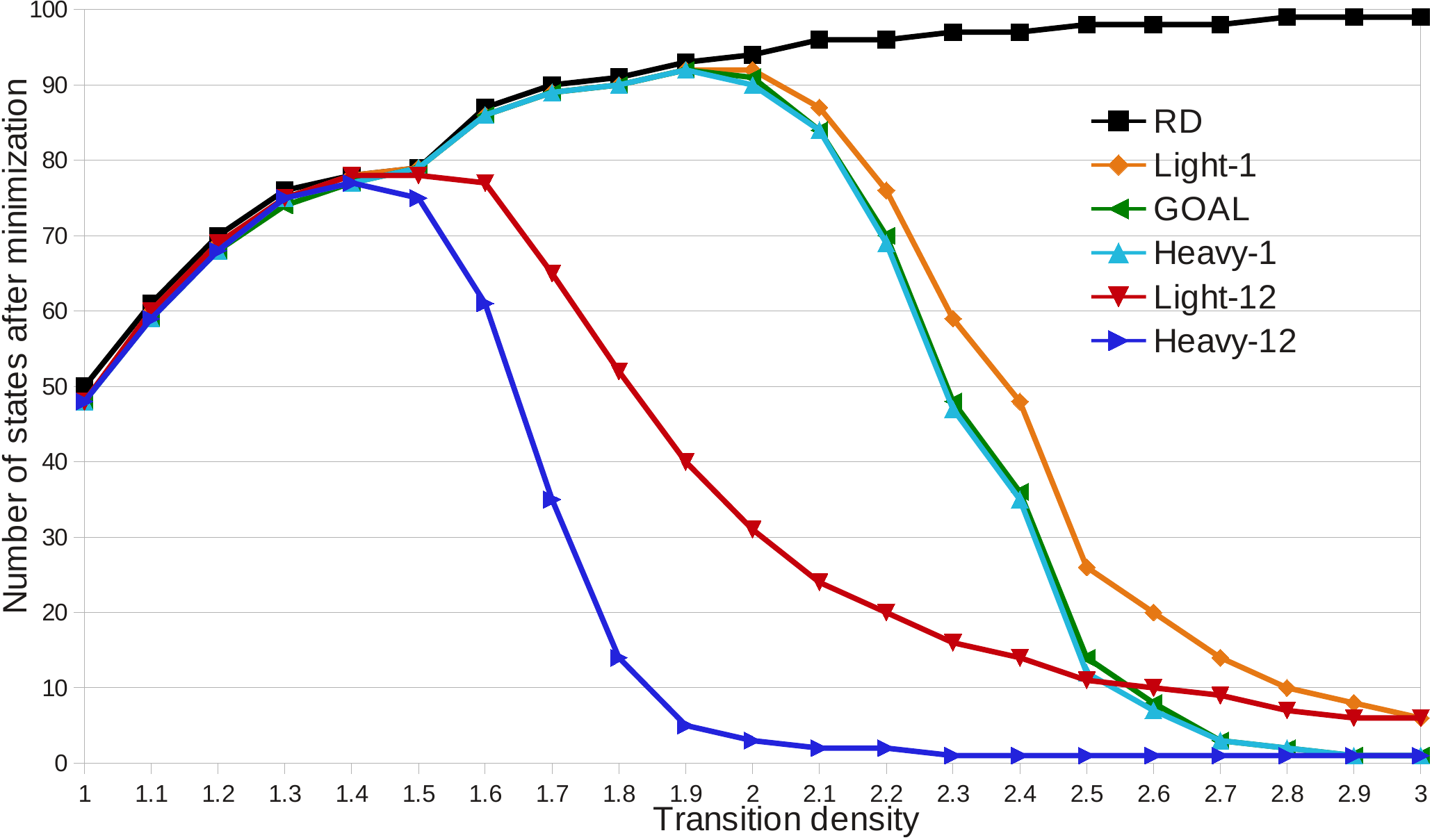}
\end{center}

\noindent
Sparse automata with low ${\it td}$ have more dead states.
For ${\it td} \le 1.4$ no technique except RD has any significant effect.
GOAL minimizes just slightly worse than Heavy-1 but it is no match for
our best techniques.
Heavy-12 vastly outperforms all others, particularly in the
interesting range between $1.4$ and $2.5$. 
Moreover, the minimization of GOAL (in particular the fair simulation minimization of
\cite{GBS02}) is very slow. For GOAL, the average minimization time per automaton varies
between 39s (at ${\it td}=1.0$) and 612s (maximal at ${\it td}=2.9$).
In contrast, for Heavy-12, the average minimization time per automaton
varies between 0.012s (at ${\it td}=1.0$) and 1.482s (max. at ${\it td}=1.7$).
So Heavy-12 minimizes not only much better, but also at least 400
times faster than GOAL (see also the scalability test).

For ${\it td} \ge 2.0$, Heavy-12 yields very small automata. Many of these are
even universal, i.e., with just one state and a universal loop.
However, this frequent universality is {\em not} due to trivial reasons 
(otherwise simpler techniques like Light-1 and GOAL would also recognize this).
Consider the following question: Given Tabakov-Vardi random automata with
parameters $n$, $|\Sigma|$ and ${\it td}$, what is the probability 
$U(n,|\Sigma|,{\it td})$ that
every state has at least one outgoing transition for every symbol in $\Sigma$?
(Such an automaton would be trivially universal if ${\it ad}=1$.)

\begin{theorem}
$U(n,|\Sigma|,{\it td}) = (\alpha(n,T)/\beta(n,T))^{|\Sigma|}$, 
with $T=n\cdot {\it td}$, 
$\alpha(n,T)=\sum_{m=n}^{n^2} {{m-n} \choose {T-n}} \sum_{i=0}^{n} (-1)^i
{n \choose i} {{m-in-1} \choose {n-1}}$
and $\beta(n,T) = {{n^2} \choose T}$
\end{theorem}
\begin{proof}
For each symbol in $\Sigma$ there are 
$T=n\cdot {\it td}$ transitions and $n^2$ possible places for transitions,
described as a grid.
$\alpha(n,T)$ is the number of ways $T$ items can be placed onto an 
$n\times n$ grid s.t. every row contains $\ge 1$ item, i.e., every state has an
outgoing transition. $\beta(n,T)$ is the number of possibilities without this
restriction, which is trivially ${{n^2} \choose T}$.
Since the Tabakov-Vardi model chooses transitions for different symbols
independently, we have $U(n,|\Sigma|,{\it td}) = (\alpha(n,T)/\beta(n,T))^{|\Sigma|}$.
It remains to compute $\alpha(n,T)$.
For the $i$-th row let $x_i \in \{1,\dots,n\}$ be the maximal column
containing an item. The remaining $T-n$ items can only be distributed to lower columns. 
Thus $\alpha(n,T) = \sum_{x_1,\dots,x_n} {{(\sum x_i)-n} \choose {T-n}}$.
With $m=\sum x_i$ and a standard dice-sum problem from \cite{Niven:1965}
the result follows.
\end{proof}

\noindent
For $n=100$, $|\Sigma|=2$ we obtain the following values for
$U(n,|\Sigma|,{\it td})$:
$10^{-15}$ for ${\it td=2.0}$, $2.9\cdot 10^{-5}$ for ${\it td=3.0}$,
$0.03$ for ${\it td=4.0}$, $0.3$ for ${\it td=5.0}$, 
$0.67$ for ${\it td=6.0}$, and $0.95$ for ${\it td=8.0}$.
So this transition saturation effect is negligible in our tested
range with ${\it td \le 3.0}$. 

While Heavy-12 performs very well, an even smaller lookahead can already
be sufficient for a good minimization. However, this depends very much on the
density ${\it td}$ of the automata. The following chart shows the effect of
the lookahead by comparing Heavy-k for varying $k$ on different classes of
random automata with different density ${\it td}=1.6, 1.7, 1.8, 1.9, 2.0$.
We have $n=100$, $|\Sigma|=2$ 
and ${\it ad}=0.5$, and every point is the average of 1000 automata.

\begin{center}
\includegraphics[scale=0.4]{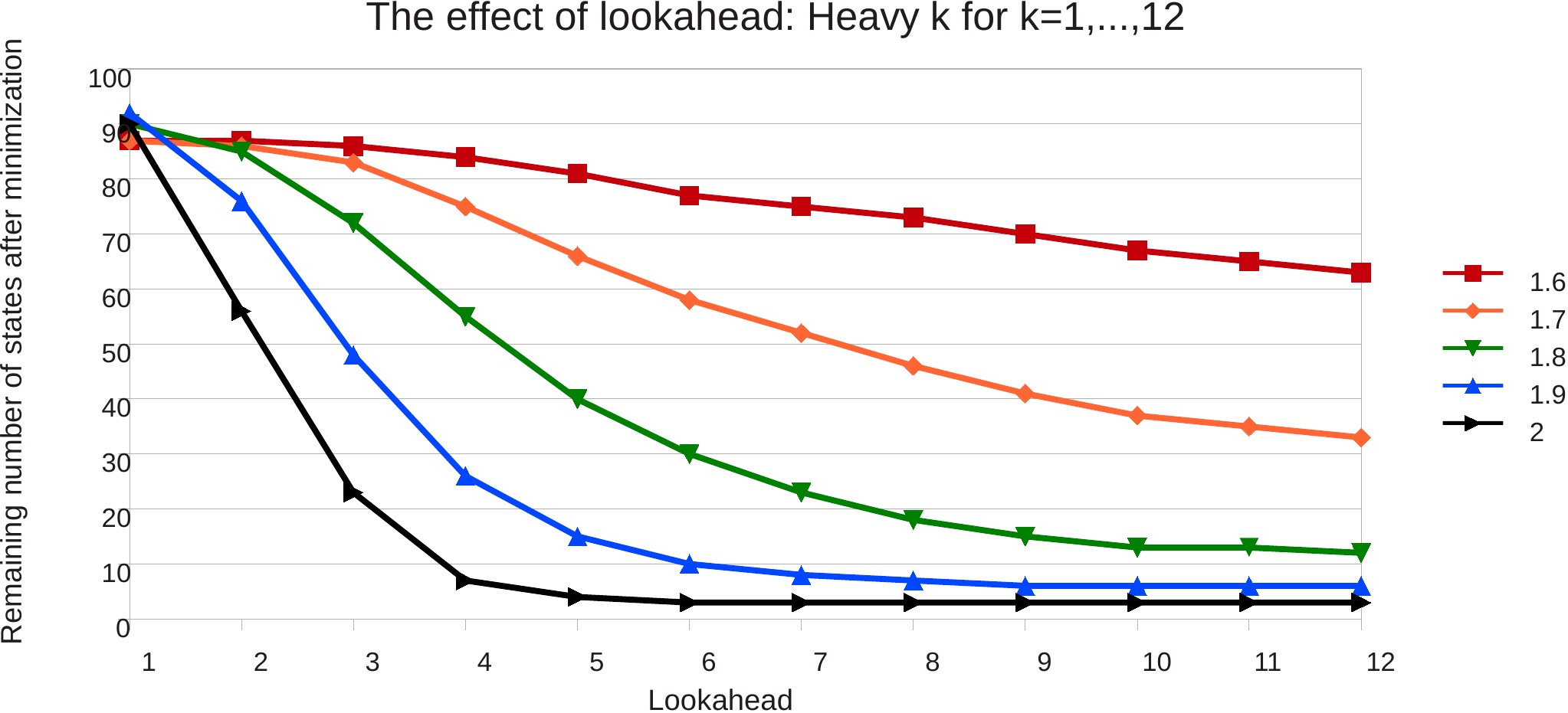}
\end{center}

\noindent 
The big advantage of Heavy-12 over Light-12 is due to the pruning techniques.
However, these only reach their full potential at higher lookaheads 
(thus the smaller difference between Heavy-1 and Light-1).
Indeed, the simulation relations get much denser with higher lookahead $k$.
We consider random automata with $n=100$, $|\Sigma|=2$ and
${\it td}=1.8$ (a nontrivial case; larger ${\it td}$ yield larger
simulations). We let ${\it ad}=0.1$ (resp. ${\it ad}=0.9$), and plot the size of fair, delayed,
direct, and backward simulation as $k$ increases from 
1 to 12. Every point is the average of 1000 automata.

\begin{center}
\includegraphics[scale=0.5]{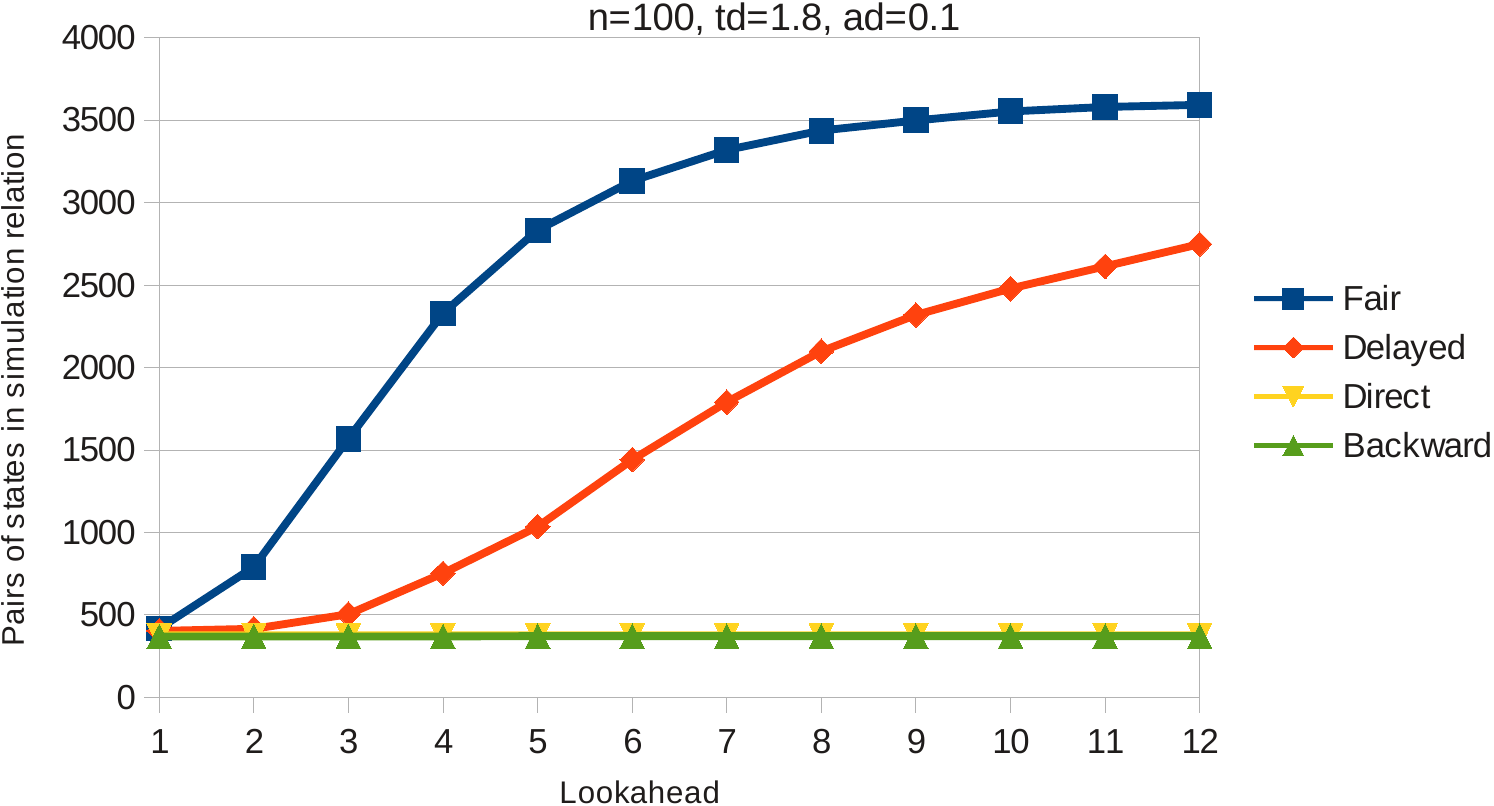}
\end{center}

\begin{center}
\includegraphics[scale=0.5]{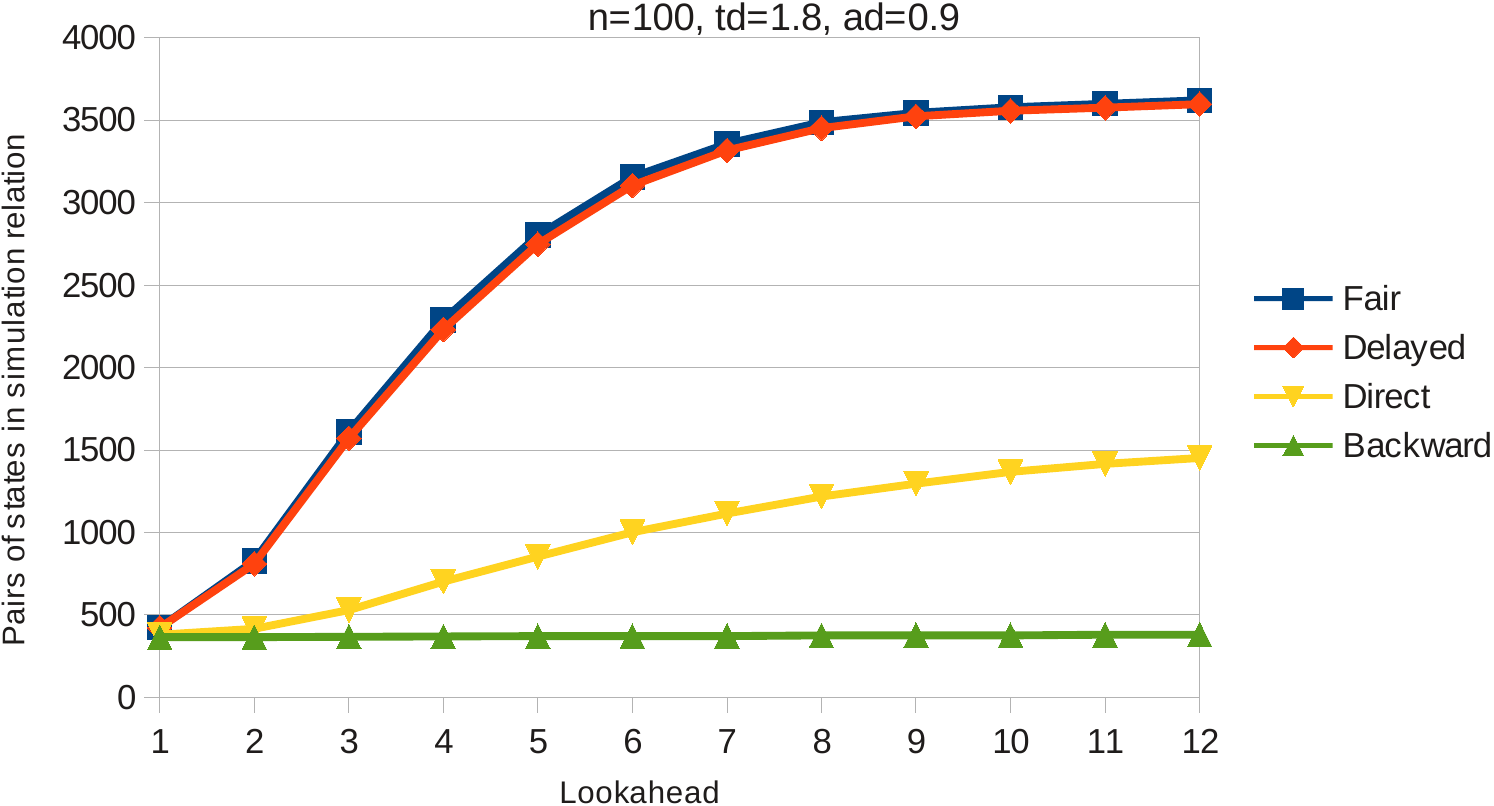}
\end{center}

\noindent
Fair/delayed simulation is not much larger than direct
simulation for $k=1$, but they benefit strongly from higher $k$.
Backward simulation increases only slightly (e.g., from 365 to 381 pairs for
${\it ad}=0.9$).
Initially, it seems as if backward/direct simulation does not benefit from higher $k$ if 
${\it ad}$ is small (on random automata), but this is wrong.
Even random automata get less random during the Heavy-k minimization process,
making lookahead more effective for backward/direct simulation.
Consider the case of $n=300$, ${\it td}=1.8$ and ${\it ad}=0.1$. 
Initially, the average ratio $|\transkdisimnumber{12}|/|\transkdisimnumber{1}|$
is $1.00036$, but after quotienting with $\transkdesimnumber{12}$
this ratio is $1.103$.

\paragraph{LTL.}

For model checking \cite{Holzmann:Spinbook}, LTL-formulae are converted into B\"uchi automata. This
conversion has been extensively studied and there are many different
algorithms which try to construct the smallest possible automaton
for a given formula (see references in \cite{GOAL_survey_paper}).
It should be noted however, that LTL is designed for human readability and 
does not cover the full class of $\omega$-regular languages. Moreover,
B\"uchi Store \cite{buechistore:2011} contains handcrafted automata for
almost every human-readable LTL-formula and none of these automata has 
more than 7 states.
Still, since many people are interested in LTL to automata conversion,
we tested how much our minimization algorithm can improve upon the best effort
of previous techniques. For LTL model checking, the size of the automata 
is not the only criterion \cite{Sebastiani-Tonetta:2003}, since 
more non-determinism also makes the problem harder. However, our 
transition pruning techniques only make an automaton `more deterministic'.

Using a function of GOAL, we created 300 random LTL-formulae of nontrivial
size: length 70, 4 predicates and probability weights 1 for boolean and 2 for
future operators. We then converted these formulae to B\"uchi automata and 
minimized them with GOAL. Of the 14 different converters implemented in GOAL
we chose LTL2BA \cite{GastinOddoux2001} (which is also used by the SPIN model
checker \cite{Holzmann:Spinbook}), since it was the only one which could handle such large
formulae. (The second best was COUVREUR which succeeded on 90\% of the
instances, but produced much larger automata than LTL2BA. The other converters 
ran out of time (4h) or memory (14GB) on most instances.)
We thus obtained 300 automata and minimized them with GOAL. 
The resulting automata vary significantly in size from 1 state to 1722 states \cite{ourtest}. 

Then we tested how much {\em further} these automata could be reduced in size by our
Heavy-12 method (cf. Appendix~\ref{sec:scalability}). In summary,
82\% of the automata could be further reduced in size. 
The average number of states declined from 138 to 78, and the average number
of transitions from 3102 to 1270. Since larger automata have a 
disproportionate effect on averages, we also computed the average reduction
ratio per automaton, i.e., 
$(1/300) \sum_{i=1}^{300} {\it newsize}_i/{\it oldsize}_i$.
(Note the difference between the average ratio and the ratio of averages.)
The average ratio was $0.76$ for states and $0.68$ for transitions.
The computation times for minimization vary a lot due to different automata sizes
(average 122s), but were always less than the time used by the LTL to automata translation. 
If one only considers the 150 automata above median size (30 states)
then the results are even stronger. 100\% of these automata could be further
reduced in size.
The average number of states declined from 267 to 149, and the average number
of transitions from 6068 to 2435. 
The average reduction ratio was $0.65$ for states and $0.54$ for transitions.
To conclude, our minimization can significantly improve the quality of
LTL to automata translation with a moderate overhead.

\paragraph{Mutual exclusion protocols.}
We consider automata derived from mutual exclusion protocols.
The protocols were described in a language of guarded commands and
automatically translated into B\"uchi automata, whose size is given in the
column `Original'. 
By row, the protocols are Bakery.1, Bakery.2, Fischer.3.1, Fischer.3.2, 
Fischer.2, Phils.1.1, Phils.2 and Mcs.1.2.
We minimize these automata with GOAL and with our Heavy-12 method 
and describe the sizes of the resulting automata and the runtime in 
subsequent columns (Java 6 on Intel i7-740, 1.73 GHz).
In some instances GOAL ran out of time (2h) or memory (14GB).

\begin{center}\scriptsize
	\begin{tabular}{|c|c|c|c|c|c|c|c|}
	  \hline
	\multicolumn{2}{|c|}{Original}
	  & \multicolumn{2}{c|}{GOAL} & Time & \multicolumn{2}{c|}{Heavy-12} & Time\\
	 \cline{1-4}\cline{6-7}
	  Trans. & States & Tr. & St. & GOAL & Tr. & St. &
	  Heavy-12\\
	  \hline
	  2597 & 1506 & N/A & N/A & $>2h$ & 696 & 477 & 6.17s\\
	  \hline
	  2085 & 1146 & N/A & N/A & $>2h$ & 927 & 643 & 9.04s\\
	  \hline
	  1401 & 638 & 14 & 10 & 15.38s & 14 & 10 & 1.16s\\
	  \hline
	  3856 & 1536 & 212 & 140 & 4529s & 96 & 70 & 5.91s\\
	  \hline
	  67590 & 21733 & N/A & N/A & oom(14GB) & 316 & 192 & 325.76s\\
	  \hline
	  464  & 161 & 362 & 134 & 540.3s & 359 & 134 & 11.51s\\
	  \hline
	  2350 & 581 & 284 & 100 & 164.2s & 225 & 97 & 4.04s\\
	  \hline
	  21509 & 7968 & 108 & 69 & 2606.7s & 95 & 62 & 48.18s\\
	  \hline
	\end{tabular}
\end{center}

\paragraph{Scalability.}
We test the scalability of Heavy-12 minimization by applying it to
Tabakov-Vardi random automata of increasing size but fixed ${\it td}$, 
${\it ad}$ and $\Sigma$. We ran four separate tests with ${\it td}=1.4, 1.6, 1.8$ and
$2.0$. In each test we fixed ${\it ad}=0.5$, $|\Sigma|=2$ and increased the number of states
from $n=50$ to $n=1000$ in increments of 50. For each parameter point 
we created 300 random automata and minimized them with
Heavy-12. We analyze the average size of the minimized automata in percent of
the original size $n$, and how the average computation time increases with
$n$.

For ${\it td}=1.4$ the average size of the minimized automata stays around 
$77\%$ of the original size, regardless of $n$.
For ${\it td}=1.6$ it stays around $65\%$.
For ${\it td}=1.8$ it {\em decreases} from
$28\%$ at $n=50$ to $2\%$ at $n=1000$.
For ${\it td}=2.0$ it {\em decreases} from
$8\%$ at $n=50$ to $<1\%$ at $n=1000$ (cf. Appendix~\ref{sec:scalability}).  
Note that the lookahead of 12 did {\em not change} with $n$.
Surprisingly, larger automata do not require larger lookahead for a good minimization.

We plot the average computation time (measured in ms) in $n$ and then compute
the optimal fit of the function ${\it time} = a*n^b$ to the data by the least-squares
method, i.e., this computes the parameters $a$ and $b$ of the function that
most closely fits the experimental data. The important parameter is the
exponent $b$. For ${\it td}=1.4, 1.6, 1.8, 2.0$ we obtain
$0.018*n^{2.14}$, $0.32*n^{2.39}$, $0.087*n^{2.05}$ and $0.055*n^{2.09}$, respectively.
Thus, the average-case complexity of Heavy-12 scales (almost) 
quadratically. This is especially surprising given that Heavy-12 does not only
compute one simulation relation but potentially many simulations until 
the repeated minimization reaches a fixpoint.
Quadratic complexity is the very best one can hope for in any method that
explicitly compares states/transitions by simulation relations, since the
relations themselves are of quadratic size. 
Lower complexity is only possible with pure partition refinement techniques 
(e.g., bisimulation, which is $O(n\log n)$), but these achieve even less
minimization than quotienting with direct simulation (i.e., next to nothing 
on hard instances).

\begin{center}
\includegraphics[scale=0.4]{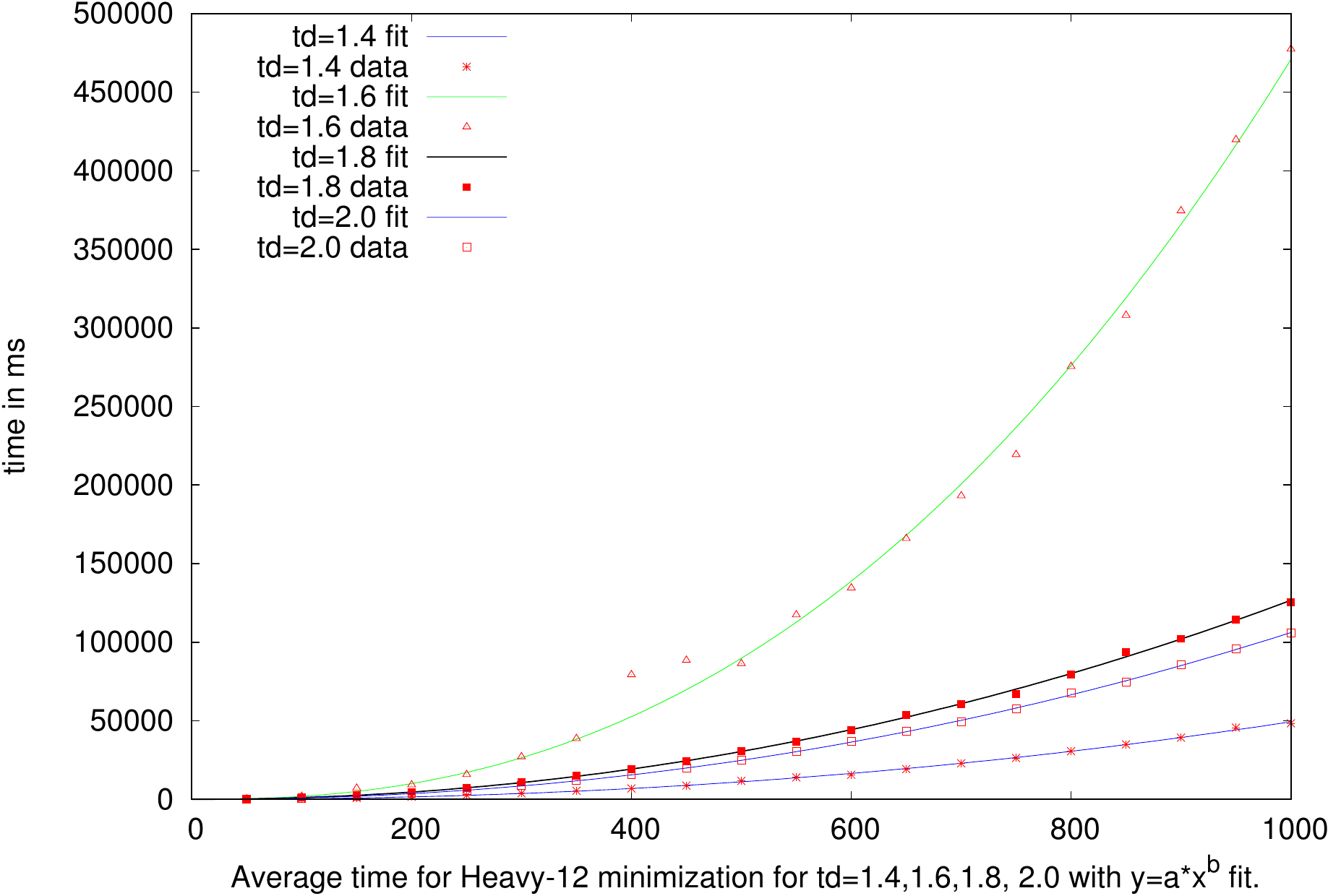}
\end{center}

\noindent
The computation time of Heavy-k depends on the class of automata, 
i.e., on the density ${\it td}$, as the scalability test above shows.
Moreover, it also depends on $k$.
The following graph shows the average computation time of Heavy-k on automata
of size 100 and varying ${\it td}$ and $k$. The most difficult cases are those where
minimization is possible (and thus the alg. does not give up quickly), but 
does not massively reduce the size of the instance. For Heavy-k, this peak is
around ${\it td}=1.6,1.7$ (like in the scalability test).

\begin{center}
\includegraphics[scale=0.4]{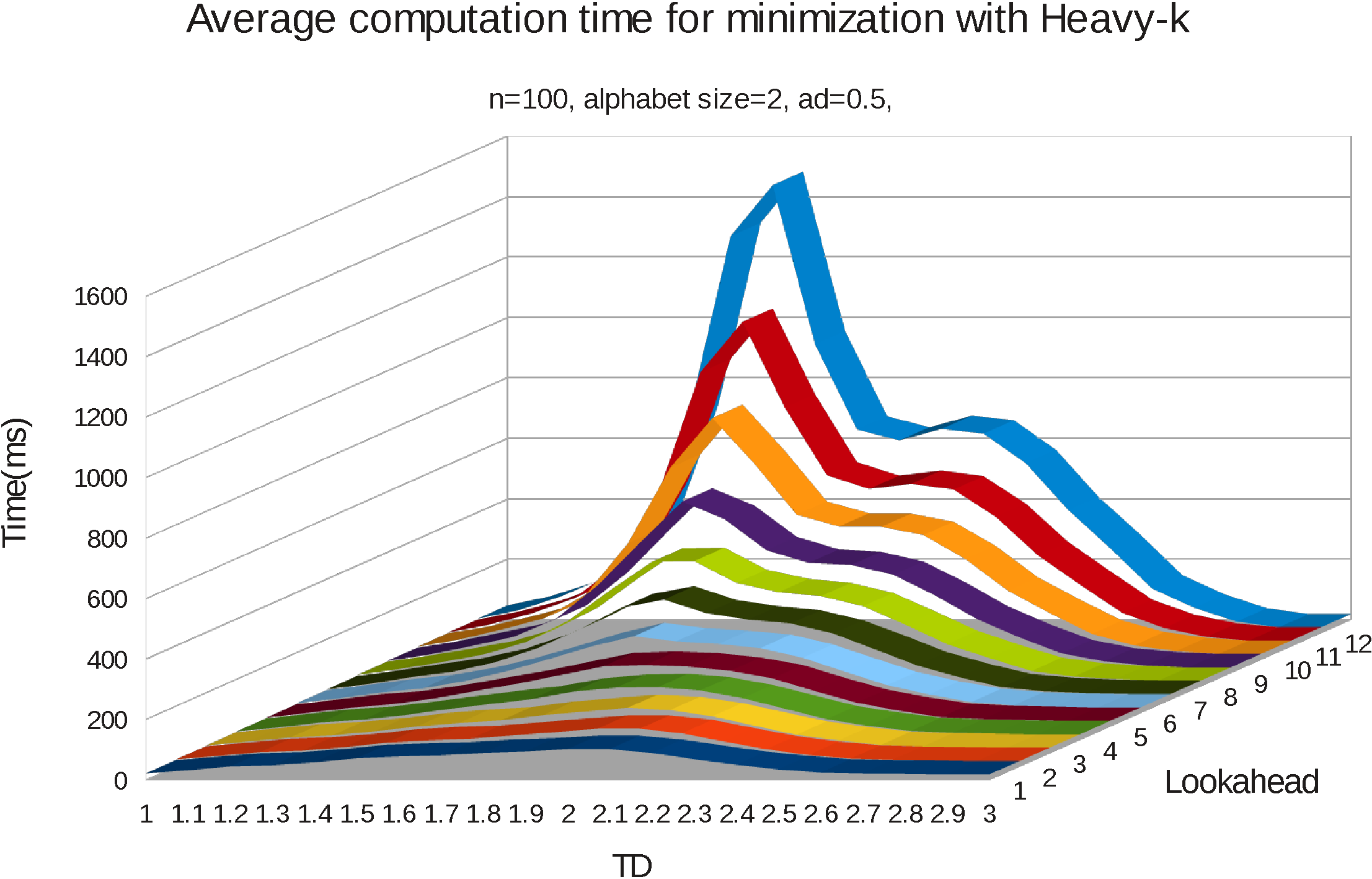}
\end{center}

\paragraph{Language Inclusion Checking.}
We test the language inclusion checking algorithm of
Section~\ref{sec:inclalg} (with lookahead up-to 15)
on nontrivial instances and compare its 
performance to previous techniques like ordinary fair simulation checking
and the best effort of GOAL (which uses simulation-based minimization 
followed by an on-the-fly variant of Piterman's construction
\cite{Pit06,GOAL_survey_paper}).
In this test we use only the polynomial time steps (1)-(3) of our algorithm,
thus it may fail in some instances.
We consider pairs of Tabakov-Vardi random automata with 1000 states each,
$|\Sigma|=2$ and ${\it ad}=0.5$. For each separate case of 
${\it td}=1.6, 1.8$ and $2.0$, we create 300 such automata pairs and check
if language inclusion holds. (For ${\it td}<1.6$ inclusion rarely holds,
except trivially if one automaton has empty language. For ${\it td}>2$
inclusion often holds but is easier to prove.)

For ${\it td}=1.6$ our algorithm solved 294 of 300 instances (i.e., 98\%):
45 included (16 in step (1) and 29 in step (2)), 249 non-included (step (3)),
and 6 failed. Average computation time 1167s.
Ordinary fair simulation solved only 13 included instances.
GOAL (timeout 60min, 14GB memory) solved only 13 included
instances (the same 13 as fair simulation) and 155 non-included instances.

For ${\it td}=1.8$ our algorithm solved 297 of 300 instances (i.e., 99\%):
104 included (103 in step (1) and 1 in step (2)) and 193 non-included 
(step (3)) and 3 failed. 
Average computation time 452s.
Ordinary fair simulation solved only 5 included instances.
GOAL (timeout 30min, 14GB memory) solved only 5 included
instances (the same 5 as fair simulation) and 115 non-included instances.

For ${\it td}=2.0$ our algorithm solved every instance:
143 included (shown in step (1)) and 157 non-included (step (3)).
Average computation time 258s.
Ordinary fair simulation solved only 1 of the 143 included instances.
GOAL (timeout 30min, 14GB memory) solved only 1 of 143 included
instances (the same one as fair simulation) and 106 of 157 non-included instances.

%% file: conclusion.tex
\section{Conclusion and Future Work}\label{sec:conclusion}

Our automata minimization techniques perform significantly better than 
previous methods. In particular, they can be applied to solve PSPACE-complete
automata problems like language inclusion for much larger instances.
While we presented our methods in the context of B\"uchi automata,
most of them trivially carry over to the simpler case of automata over finite
words.
Future work includes more efficient algorithms for computing lookahead
simulations, either along the lines of \cite{HHK:FOCS95} for normal simulation,
or by using symbolic representations of the relations.
Moreover, we are applying similar techniques to minimize tree-automata.

%% file: appendix.tex
\section{Additional Experiments}

This appendix contains additional material related to experiments with our minimization algorithm
(cf. Section~\ref{sec:experiments}).

\subsection{The Effect of the Acceptance Density}
\label{sec:acceptance_density}

Figure~\ref{fig:acceptance_density} shows the performance of our
minimization algorithm on random automata with acceptance density $0.5$ and
$0.1$, respectively.
Clearly, variations in the acceptance density do not affect our methods 
with lookahead (e.g., Light-12 and Heavy-12) very much.
However, a small acceptance density like $0.1$ makes the problem somewhat
harder for methods without lookahead (e.g., Light-1 and Heavy-1).

\begin{figure}[h]\begin{center}
\includegraphics[scale=0.5]{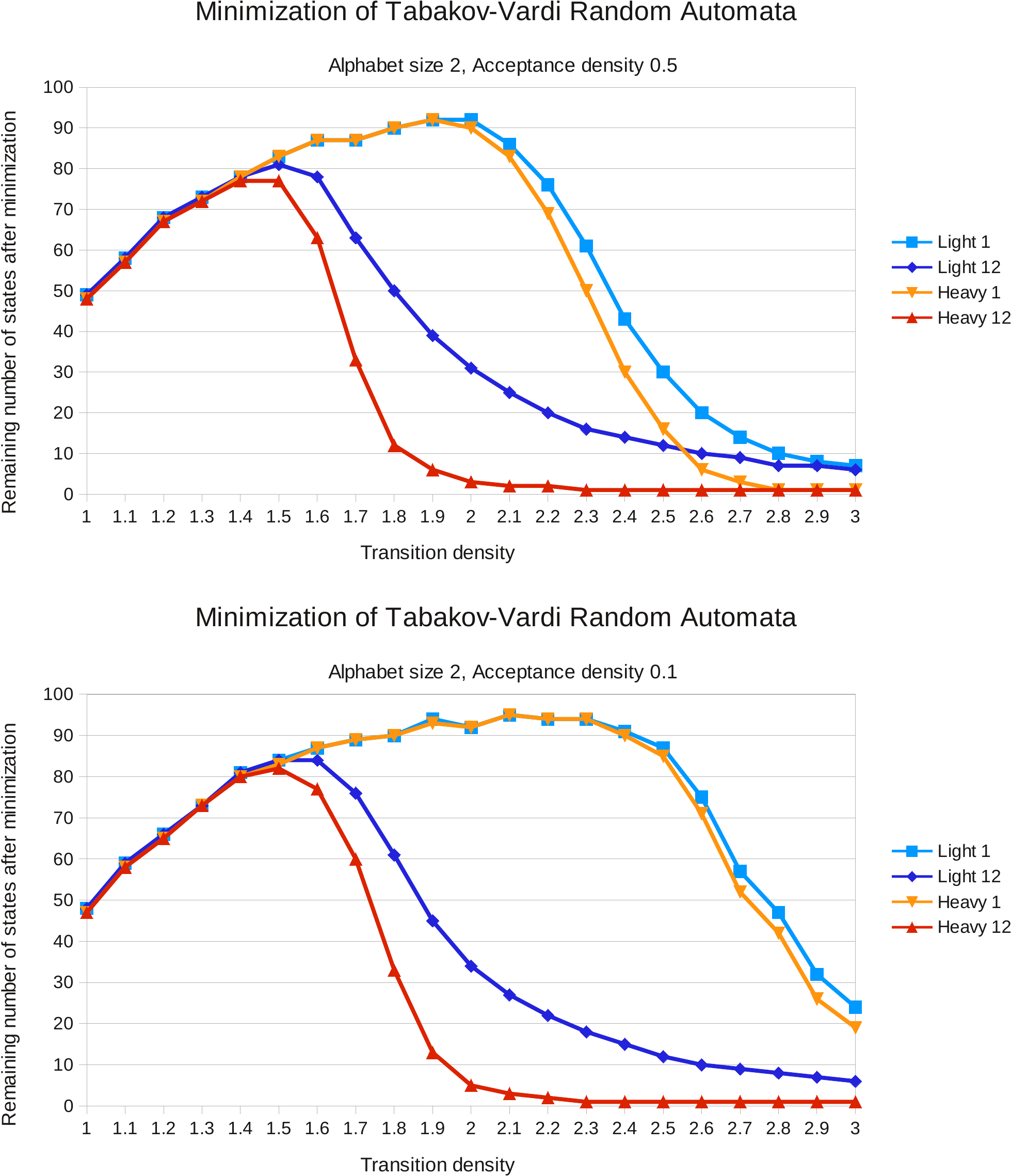}
\end{center}
\caption{Minimization of Tabakov-Vardi random automata with $n=100$,
$|\Sigma|=2$, ${\it ad}=0.5$ (top), ${\it ad}=0.1$ (bottom) and varying
${\it td}$. 
We use the Light 1, Light 12, Heavy 1 and Heavy 12 methods
and plot the average number of states of the minimized automata. Every point
in the top (resp. bottom) graph the average
of 1000 (resp. 300) automata.
Note how a small acceptance density makes minimization harder without lookahead, but
not much harder for lookahead 12.
}\label{fig:acceptance_density}
\end{figure}

\newpage
\subsection{Scalability}
\label{sec:scalability}

In this section we present the complete data for our scalability experiments.
We tested our Heavy-12 minimization algorithm on random automata of increasing size but fixed ${\it td}$, ${\it ad}$ and $\Sigma$.
In Figure~\ref{fig:scalability1} we show the reduction in size,
while in Figure~\ref{fig:scalability2} we show the computation time (for the same set of experiments).

\begin{figure}[h]
\begin{center}
\includegraphics[scale=0.5]{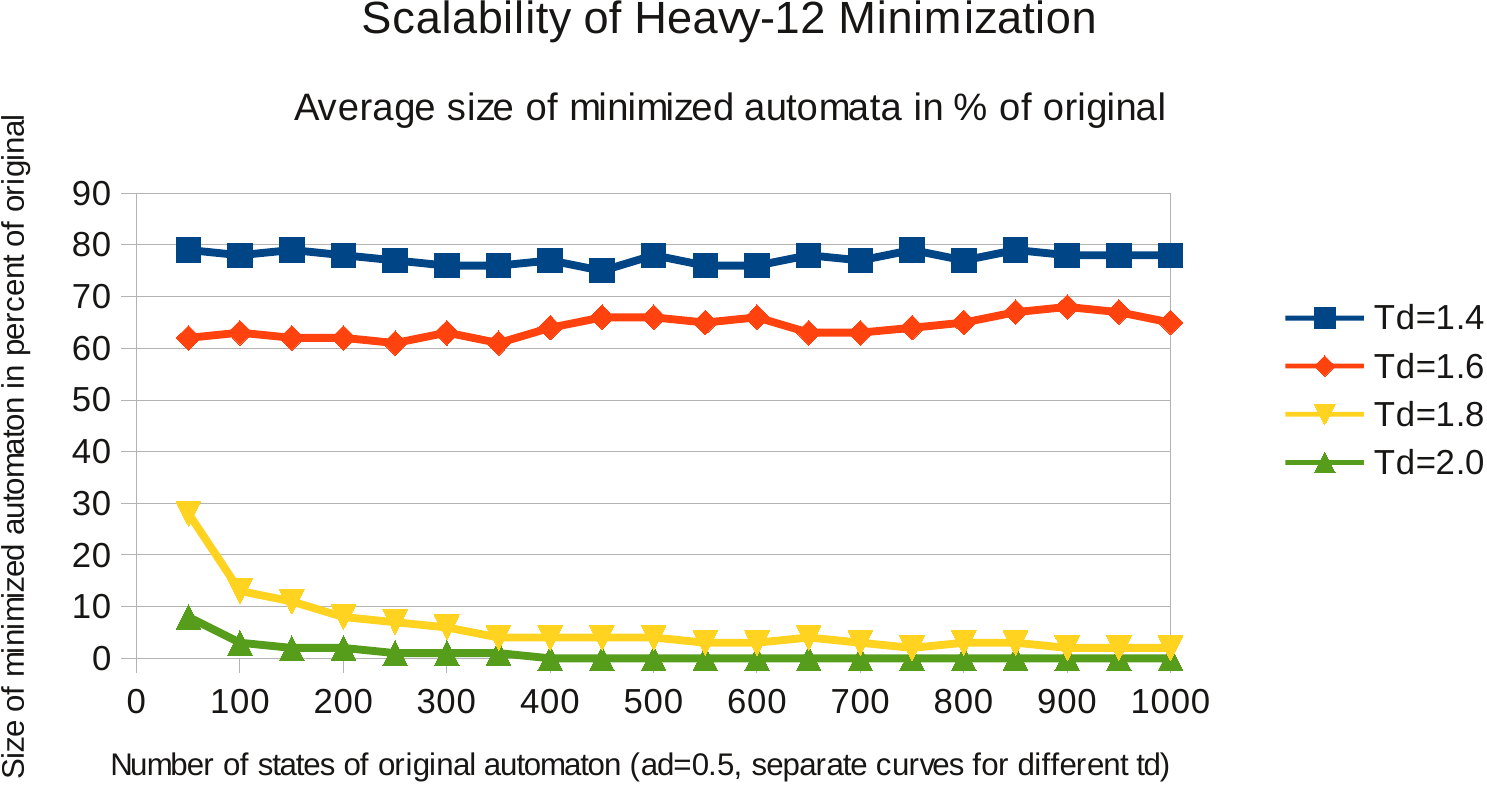}
\end{center}
\caption{Minimization of Tabakov-Vardi random automata with ${\it ad}=0.5$,
$|\Sigma|=2$, and increasing $n=50,100,\dots,1000$.
Different curves for different ${\it td}$.
We plot the average size of the Heavy-12 minimized automata, in percent of
their original size. Every point is the average of 300 automata.
Note that the lookahead of 12 does not change, i.e., larger automata do not
require a higher lookahead for a good minimization.
}\label{fig:scalability1}
\end{figure}

\begin{figure}[h]
\begin{center}
\includegraphics[scale=0.45]{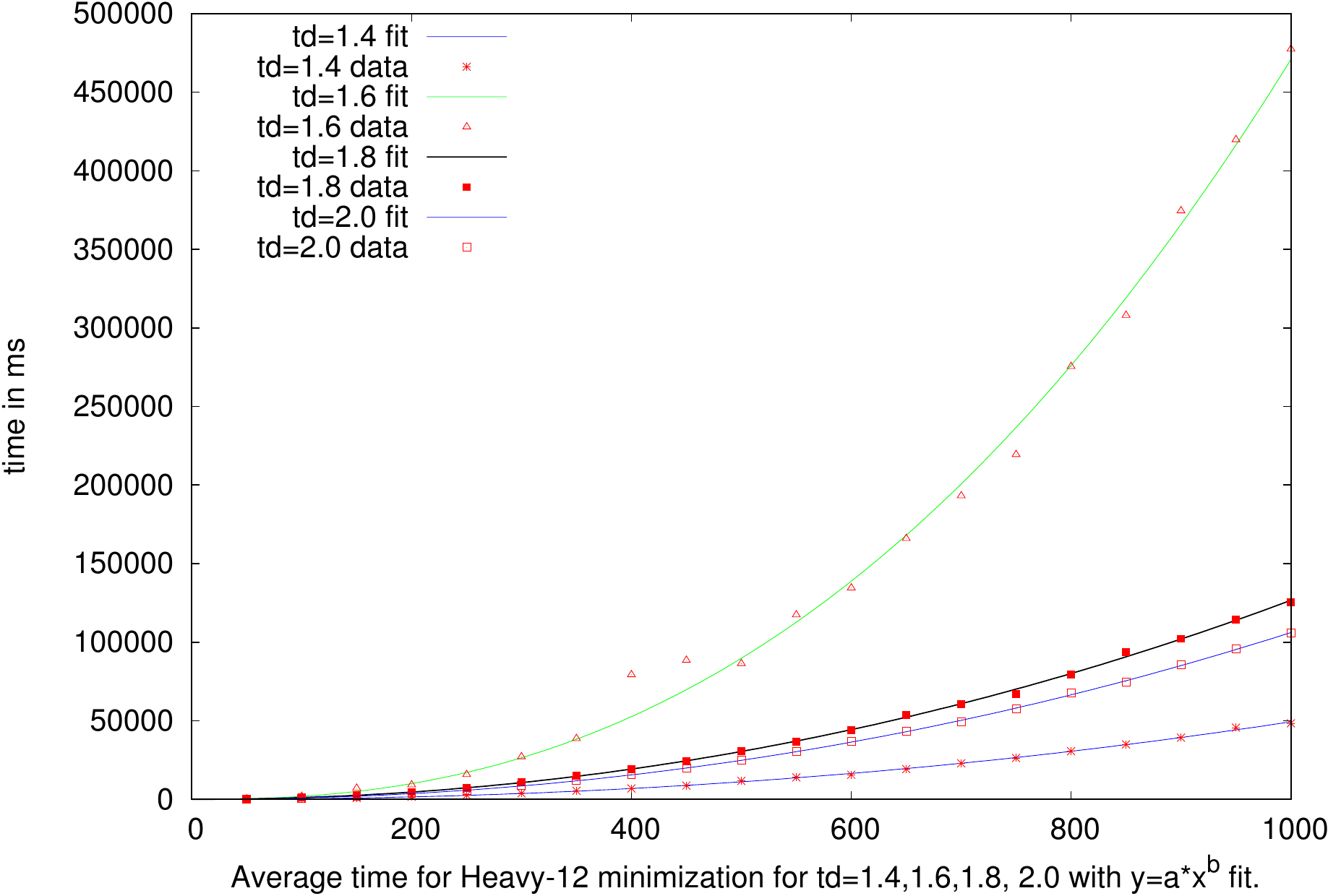}
\end{center}
\caption{Minimization of Tabakov-Vardi random automata as in Figure~\ref{fig:scalability1}.
Here we plot the average computation time (in ms) for the minimization,
and a least-squares fit by the function $a*n^b$.
For ${\it td}=1.4, 1.6, 1.8, 2.0$ we obtain
$0.018*n^{2.14}$, $0.32*n^{2.39}$, $0.087*n^{2.05}$ and $0.055*n^{2.09}$, respectively.
}\label{fig:scalability2}
\end{figure}

\newpage
\section{Non-transitivity of Lookahead Simulation}
\label{sec:non_transitivity}

\input{fig_non_transitive}

In this section we show that lookahead simulation is not transitive for $k \geq 2$.
Consider the example in Figure~\ref{fig:lookhead_non_transitive}.
We have $p_0 \ksim k q_0 \ksim k r_0$ (and $k = 2$ suffices),
but $p_0 \not \ksim k r_0$ for any $k > 0$. In fact,

\begin{itemize}
	\item $p_0 \ksim k q_0$, with $k = 2$:
	Duplicator takes the transition via $q_1$ or $q_2$
	depending on whether Spoiler plays word $(a+b)a$ or $(a+b)b$,
	respectively.
	\item $q_0 \ksim k r_0$, with $k = 2$:
	If Spoiler goes to $q_1$ or $q_2$, then Duplicator goes to $r_1$ or $r_2$, respectively.
	That $q_1 \ksim k r_1$ holds can be shown as follows
	(the case $q_2 \ksim k r_2$ is similar).
	If Spoiler takes transitions $q_1 \goesto a q_0 \goesto a q_1$,
	then Duplicator does $r_1 \goesto a r_1 \goesto a r_1$,
	and	if Spoiler does $q_1 \goesto a q_0 \goesto b q_1$,
	then Duplicator does $r_1 \goesto a r_2 \goesto b r_1$.
	The other cases are similar.
	\item $p_0 \not \ksim k r_0$, for any $k > 0$.
	From $r_0$,	Duplicator can play a trace for any word $w$ of length $k > 0$,
	but in order to extend it to a trace of length $k + 1$ for any $w' = wa$ or $wb$,
	she needs to know whether the last $(k + 1)$-th symbol is $a$ or $b$.
	Thus, no finite lookahead suffices for Duplicator.
\end{itemize}
Incidentally, notice that $r_0$ simulates $p_0$ with $k$-continuous simulation,
and $k = 2$ suffices.

As shown in Section~\ref{sec:lookahead}, non-transitivity of lookahead simulation
is not an obstacle to its applications. Since it is only used to compute 
good under-approximations of certain preorders, one can simply consider its
transitive closure (which is easily computed).

\newpage
\section{Fixpoint Logic Characterization of Lookahead Simulation}
\label{sec:fixedpoint}

In this section we give a fixpoint logic characterization of lookahead
simulation, using the modal $\mu$-calculus.
Basically it follows from the following preservation property enjoyed by lookahead simulation:
Let $x \in \{ \mathrm{di, de, f, bw} \}$ and $k > 0$.
When Duplicator plays according to a winning strategy,
in any configuration $(p_i, q_i)$ of the resulting play, $p_i \kxsim k x q_i$.
Thus, $k$-lookahead simulation (without acceptance condition) can be characterized as the largest $X \subseteq Q \times Q$
which is closed under a certain monotone predecessor operator.
For convenience, we take the point of view of Spoiler,
and compute the complement relation $W^x = (Q \times Q) \setminus \kxsim k x$ instead.
This is particularly useful for delayed simulation,
since we can avoid recording the obligation bit (see
\cite{etessami:etal:fairsimulations:05})
by using the technique of \cite{piterman:generalized06}.

\subparagraph{Direct and backward simulation.}
Consider the following predecessor operator $\cpredi X$, for any set $X \subseteq Q \times Q$:
\begin{align*}
	\cpredi X &= \{ (p_0, q_0) \st
		\exists(p_0 \goesto {a_0} p_1 \goesto {a_1} \cdots \goesto {a_{k-1}} p_k)  \\
								& \forall (q_0 \goesto {a_0} q_1 \goesto {a_1} \cdots \goesto {a_{m-1}} q_m), 0 < m \leq k, \\
	\textrm{\it either} \quad	& \exists (0 \leq j \leq m) \cdot p_j \in F \textrm{ and } q_j \not\in F, \\
	\textrm{\it or}		\quad 	& (p_m, q_m) \in X \}
\end{align*}
Intuitively, $(p,q) \in \cpredi X$ iff, from position $(p, q)$,
in one round of the game Spoiler can either force the game in $X$,
or violate the winning condition for direct simulation.
For backward simulation, $\cprebw X$ is defined analogously,
except that transitions are reversed and also initial states are taken into account:
\begin{align*}
	\cprebw X &= \{ (p_0, q_0) \st
		\exists(p_0 \comesfrom {a_0} p_1 \comesfrom {a_1} \cdots \comesfrom {a_{k-1}} p_k) \\
				& \forall (q_0 \comesfrom {a_0} q_1 \comesfrom {a_1} \cdots \comesfrom {a_{m-1}} q_m), 0 < m \leq k, \\
	\textrm{\it either} \quad	& \exists (0 \leq j \leq m) \cdot p_j \in F \textrm{ and } q_j \not\in F, \\
	\textrm{\it or} 	\quad	& \exists (0 \leq j \leq m) \cdot p_j \in I \textrm{ and } q_j \not\in I, \\
	\textrm{\it or}		\quad 	& (p_m, q_m) \in X \}
\end{align*}
\begin{remark}\label{rem:no_deadlocks}
	The definition of $\cprex x X$ requires that the automaton has no deadlocks;
	otherwise, Spoiler would incorrectly lose if she can only perform at most $k' < k$ transitions.
	We assumed that the automaton is complete to keep the definition simple,
	but our implementation works with general automata.

        Intuitively, the generalization to incomplete automata works as
        follows. If Spoiler's move reaches a deadlocked state after $k'$ steps,
        where $1 \le k' < k$ then Spoiler does not immediately lose. Instead
        Duplicator needs to reply to this move of length $k'$.
        In other words, if Spoiler's move ends in a deadlocked state then the
        lookahead requirements are weakened, because one simply cannot demand
        any more steps from Spoiler.
\end{remark}
\noindent
For $X = \emptyset$, $\cprex x X$ is the set of states from which Spoiler wins in at most one step.
Thus, Spoiler wins iff she can eventually reach $\cprex x \emptyset$.
Formally, for $x \in \{\mathrm{di, bw}\}$, \[ W^x = \mu W \cdot \cprex x W \nonumber \]

\ignore{
	\begin{remark}\label{rem:early_stop}
		The computation of $\cpredi X$ can be speeded up by noting that,
		if Spoiler fails to force the game in $X$ with lookahead $k$,
		then she will also fail with lookahead $k+1$.
		In this way, $(p,q)\not\in\cpredi X$ can be determined by constructing the attack of Spoiler incrementally,
		starting with length 1, 2, \dots, and stopping as soon as Duplicator has a defense to drive the computation outside $X$.
		If $(p,q) \in \cpredi X$, then we clearly have to consider attacks of Spoiler up to length $k$.
		Similarly, also Duplicator's defense can be built incrementally,
		stopping as soon as $(p_m,q_m) \not \in \cpredi X$ for some $m \leq k$.
	\end{remark}
}

\subparagraph{Delayed and fair simulation.}

We introduce a more elaborate three-arguments predecessor operator $\cprelong X Y Z$.
Intuitively, a configuration belongs to $\cprelong X Y Z$ iff
Spoiler can make a move s.t., for any Duplicator's reply,
at least one of the following conditions holds:
\begin{enumerate}
	\item Spoiler visits an accepting state, while Duplicator never does so; then, the game goes to $X$.
	\item Duplicator never visits an accepting state; the game goes to $Y$.
	\item The game goes to $Z$ (without any further condition).
\end{enumerate}

%
%
\begin{align}
	\nonumber
	\cprelong X Y Z &= \{ (p_0, q_0) \st
		\exists(p_0 \goesto {a_0} p_1 \goesto {a_1} \cdots \goesto {a_{k-1}} p_k) \\ \nonumber	
		&\forall (q_0 \goesto {a_0} q_1 \goesto {a_1} \cdots \goesto {a_{m-1}} q_m)\ \cdot \forall (0 < m \leq k) \cdot \\ \nonumber
		\textrm{\it either} \quad	&	\exists (0 \leq i \leq m) \cdot p_i \in F,
										\forall (0 \leq j \leq m) \cdot q_j \not\in F,
										(p_m, q_m) \in X \\ \nonumber
		\textrm{\it or}		\quad	&	\forall (0 \leq j \leq m) \cdot q_j \not\in F,
										(p_m, q_m) \in Y \\
		\textrm{\it or}		\quad	&	(p_m, q_m) \in Z \} \nonumber
\end{align}
%

For fair simulation, Spoiler wins iff, except for finitely many rounds,
she visits accepting states infinitely often while Duplicator does not visit any accepting state at all.
Thus, \[ W^\mathrm f = \mu Z \cdot \nu X \cdot \mu Y \cdot \cprelong X Y Z \nonumber \]

For delayed simulation, Spoiler wins if, after finitely many rounds, the following conditions are both satisfied:
\begin{inparaenum}[1)]
	\item She can visit an accepting state, and
	\item She can prevent Duplicator from visiting accepting states in the future.
\end{inparaenum}
For condition 1), let $\cpreone X Y := \cprelong X Y Y$,
and, for 2), $\cpretwo X Y := \cprelong X X Y$. Then,
\[ W^\mathrm {de} = \mu W \cdot \cpreone {\nu X \cdot \cpretwo X W} W \nonumber \]

\ignore{
The evaluation of $\cprelong X Y Z$ can be optimized by exploring the search
space of the $\exists(p = p_0 \goesto {a_0} p_1 \goesto {a_1} \cdots \goesto
{a_{k-1}} p_k)$ quantification in the following way.
Consider a prefix $p = p_0 \goesto {a_0} p_1 \goesto {a_1} \cdots \goesto
{a_{k'-1}} p_{k'}$ for some $k' < k$. If the rest of the condition (with $k'$
substituted for $k$) is false even for this prefix, then it is also false for
all extensions of $p = p_0 \goesto {a_0} p_1 \goesto {a_1} \cdots \goesto
{a_{k'-1}} p_{k'}$, which thus do not need to be considered.
On the other hand one can stop when an instance 
$p_0 \goesto {a_0} p_1 \goesto {a_1} \cdots \goesto {a_{k-1}} p_k$ with full
$k$ yields true.
Moreover, one explores the search space of the 
$\forall (0 < m \leq k) \forall (q = q_0 \goesto {a_0} q_1 \goesto {a_1}
\cdots \goesto {a_{m-1}} q_m)$
quantification in order of increasing $m$ and stops as soon as an instance
yields false.
}

%% file: fig_non_transitive.tex
\begin{figure} \centering
	\begin{tikzpicture}[on grid, node distance=2cm and 2cm]
		\tikzstyle{vertex} = [state]
				
		\path node [vertex] (p0) {$p_0$};
		
		\path node [vertex] (q0) [right = 3cm of p0] {$q_0$};
		\path node [vertex] (q1) [below left = 1.5cm and 1cm of q0] {$q_1$};
		\path node [vertex] (q2) [below right = 1.5cm and 1cm of q0] {$q_2$};
		
		\path node [vertex] (r0) [right = 4cm of q0] {$r_0$};
		\path node [vertex] (r1) [below left = 1.5cm and 1cm of r0] {$r_1$};
		\path node [vertex] (r2) [below right = 1.5cm and 1cm of r0] {$r_2$};
		
		\path[->]
		
			(p0) edge [loop above] node {$a, b$} ()
			
			(q0) edge node [left] {$a,b$} (q1)
			(q0) edge node [right] {$a,b$} (q2)
			(q1) edge [bend left = 60] node [above left] {$a$} (q0)
			(q2) edge [bend right = 60] node [above right] {$b$} (q0)
			
			(r0) edge node [left] {$a,b$} (r1)
			(r0) edge node [right] {$a,b$} (r2)
			(r1) edge [loop below] node {$a$} ()
			(r1) edge [bend left = 20] node [above] {$a$} (r2)
			(r2) edge [loop below] node {$b$} ()
			(r2) edge [bend left = 20] node [below] {$b$} (r1);
			
		\path
			(p0) -- node [right = -.4cm of p0] {$\ksim k$} (q0)
			(q0) -- node [right = .2cm of q0] {$\ksim k$} (r0);
		
		\begin{pgfonlayer}{background}
		\end{pgfonlayer}
			
	\end{tikzpicture}
	
	\caption{Lookahead simulation is not transitive.}
	
	\label{fig:lookhead_non_transitive}
	
\end{figure}
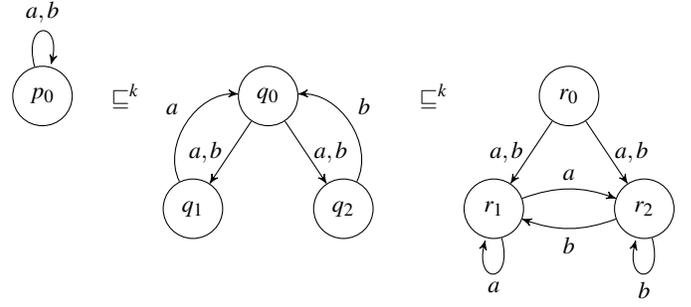

%% file: ROOT.bbl
\begin{thebibliography}{35}
\providecommand{\natexlab}[1]{#1}
\providecommand{\url}[1]{\texttt{#1}}
\expandafter\ifx\csname urlstyle\endcsname\relax
  \providecommand{\doi}[1]{doi: #1}\else
  \providecommand{\doi}{doi: \begingroup \urlstyle{rm}\Url}\fi

\bibitem[RAB()]{RABIT}
RABIT tool: www.languageinclusion.org/doku.php?id=tools.

\bibitem[our()]{ourtest}
See the appendix of this technical report for details, and
www.languageinclusion.org/doku.php?id=tools for the Java code of the 
tools and the data of the experiments.

\bibitem[Abdulla et~al.(2010{\natexlab{a}})Abdulla, Chen, Clemente, Holik,
  Hong, Mayr, and Vojnar]{abdulla:simulationsubsumption}
P.~Abdulla, Y.-F. Chen, L.~Clemente, L.~Holik, C.-D. Hong, R.~Mayr, and
  T.~Vojnar.
\newblock {Simulation Subsumption in Ramsey-Based B\"{u}chi Automata
  Universality and Inclusion Testing}.
\newblock In T.~Touili, B.~Cook, and P.~Jackson, editors, \emph{Computer Aided
  Verification}, volume 6174 of \emph{LNCS}, pages 132--147,
  2010{\natexlab{a}}.
\newblock ISBN 978-3-642-14294-9.
\newblock \doi{10.1007/978-3-642-14295-6\_14}.
\newblock URL \url{http://dx.doi.org/10.1007/978-3-642-14295-6\_14}.

\bibitem[Abdulla et~al.(2011)Abdulla, Chen, Clemente, Holik, Hong, Mayr, and
  Vojnar]{Rabit_CONCUR2011}
P.~Abdulla, Y.-F. Chen, L.~Clemente, L.~Holik, C.-D. Hong, R.~Mayr, and
  T.~Vojnar.
\newblock {Advanced Ramsey-based {B\"uchi} Automata Inclusion Testing}.
\newblock In J.-P. Katoen and B.~K{\"o}nig, editors, \emph{International
  Conference on Concurrency Theory}, volume 6901 of \emph{LNCS}, pages
  187--202, Sept. 2011.

\bibitem[Abdulla et~al.(2009)Abdulla, Chen, Hol\'{\i}k, and
  Vojnar]{AbdullaCHV09}
P.~A. Abdulla, Y.-F. Chen, L.~Hol\'{\i}k, and T.~Vojnar.
\newblock Mediating for reduction (on minimizing alternating {B{\"u}chi}
  automata).
\newblock In \emph{FSTTCS}, volume~4 of \emph{LIPIcs}, pages 1--12. Schloss
  Dagstuhl - Leibniz-Zentrum fuer Informatik, 2009.

\bibitem[Abdulla et~al.(2010{\natexlab{b}})Abdulla, Chen, Holik, Mayr, and
  Vojnar]{antichain:NFA:improved:10}
P.~A. Abdulla, Y.-F. Chen, L.~Holik, R.~Mayr, and T.~Vojnar.
\newblock {When Simulation Meets Antichains}.
\newblock In \emph{Tools and Algorithms for the Construction and Analysis of
  Systems}, volume 6015 of \emph{LNCS}, 2010{\natexlab{b}}.
\newblock URL \url{http://hal.inria.fr/inria-00460294/en/}.

\bibitem[Bustan and Grumberg(2003)]{simulationminimization:03}
D.~Bustan and O.~Grumberg.
\newblock Simulation-based minimization.
\newblock \emph{ACM Trans. Comput. Logic}, 4:\penalty0 181--206, April 2003.
\newblock ISSN 1529-3785.
\newblock \doi{http://doi.acm.org/10.1145/635499.635502}.
\newblock URL \url{http://doi.acm.org/10.1145/635499.635502}.

\bibitem[Clemente(2011)]{buchiquotient:ICALP11}
L.~Clemente.
\newblock {B\"{u}chi Automata Can Have Smaller Quotients}.
\newblock In L.~Aceto, M.~Henzinger, and J.~Sgall, editors, \emph{ICALP},
  volume 6756 of \emph{LNCS}, pages 258--270. 2011.
\newblock ISBN 978-3-642-22011-1.
\newblock \doi{10.1007/978-3-642-22012-8\_20}.
\newblock URL \url{http://arxiv.org/pdf/1102.3285}.

\bibitem[Clemente(2012)]{Clemente:PhD}
L.~Clemente.
\newblock \emph{Generalized Simulation Relations with Applications in Automata
  Theory}.
\newblock PhD thesis, University of Edinburgh, 2012.

\bibitem[Dill et~al.(1991)Dill, Hu, and Wont-Toi]{dill:inclusion:1992}
D.~L. Dill, A.~J. Hu, and H.~Wont-Toi.
\newblock {Checking for Language Inclusion Using Simulation Preorders}.
\newblock In \emph{Computer Aided Verification}, volume 575 of \emph{LNCS}.
  Springer-Verlag, 1991.
\newblock \doi{10.1007/3-540-55179-4\_25}.
\newblock URL \url{http://dx.doi.org/10.1007/3-540-55179-4\_25}.

\bibitem[Doyen and Raskin(2010)]{doyen:raskin:antichains}
L.~Doyen and J.-F. Raskin.
\newblock {Antichains Algorithms for Finite Automata}.
\newblock In \emph{Tools and Algorithms for the Construction and Analysis of
  Systems}, volume 6015 of \emph{LNCS}, pages 2--22. Springer-Verlag, 2010.

\bibitem[Etessami(2002)]{etessami:hierarchy02}
K.~Etessami.
\newblock {A Hierarchy of Polynomial-Time Computable Simulations for Automata}.
\newblock In \emph{International Conference on Concurrency Theory}, volume 2421
  of \emph{LNCS}, pages 131--144. Springer-Verlag, 2002.
\newblock \doi{10.1007/3-540-45694-5\_10}.
\newblock URL \url{http://dx.doi.org/10.1007/3-540-45694-5\_10}.

\bibitem[Etessami and Holzmann(2000)]{optimizing:concur2000}
K.~Etessami and G.~Holzmann.
\newblock {Optimizing {B{\"{u}}chi} Automata}.
\newblock In \emph{International Conference on Concurrency Theory}, volume 1877
  of \emph{LNCS}, pages 153--168. Springer-Verlag, 2000.

\bibitem[Etessami et~al.(2005)Etessami, Wilke, and
  Schuller]{etessami:etal:fairsimulations:05}
K.~Etessami, T.~Wilke, and R.~A. Schuller.
\newblock {Fair Simulation Relations, Parity Games, and State Space Reduction
  for B\"{u}chi Automata}.
\newblock \emph{SIAM J. Comput.}, 34\penalty0 (5):\penalty0 1159--1175, 2005.
\newblock \doi{10.1137/S0097539703420675}.
\newblock URL \url{http://epubs.siam.org/sam-bin/dbq/article/42067}.

\bibitem[Fogarty and Vardi(2009)]{seth:buchi}
S.~Fogarty and M.~Vardi.
\newblock {B\"{u}chi Complementation and Size-Change Termination}.
\newblock In S.~Kowalewski and A.~Philippou, editors, \emph{Tools and
  Algorithms for the Construction and Analysis of Systems}, volume 5505 of
  \emph{LNCS}, pages 16--30. 2009.
\newblock \doi{10.1007/978-3-642-00768-2\_2}.
\newblock URL \url{http://dx.doi.org/10.1007/978-3-642-00768-2\_2}.

\bibitem[Fogarty and Vardi(2010)]{seth:efficient}
S.~Fogarty and M.~Y. Vardi.
\newblock {Efficient {B\"{u}chi} Universality Checking.}
\newblock In \emph{Tools and Algorithms for the Construction and Analysis of
  Systems}, pages 205--220, 2010.

\bibitem[Fogarty et~al.(2011)Fogarty, Kupferman, Vardi, and
  Wilke]{fogarty_et_al:LIPIcs:2011:3235}
S.~Fogarty, O.~Kupferman, M.~Y. Vardi, and T.~Wilke.
\newblock {Unifying {B{\"u}chi} Complementation Constructions}.
\newblock In M.~Bezem, editor, \emph{Computer Science Logic}, volume~12 of
  \emph{LIPIcs}, pages 248--263. Schloss Dagstuhl--Leibniz-Zentrum fuer
  Informatik, 2011.
\newblock \doi{http://dx.doi.org/10.4230/LIPIcs.CSL.2011.248}.

\bibitem[Gastin and Oddoux(2001)]{GastinOddoux2001}
P.~Gastin and D.~Oddoux.
\newblock Fast {LTL} to {B\"uchi} automata translation.
\newblock In \emph{CAV}, volume 2102 of \emph{LNCS}, pages 53--65. Springer,
  2001.

\bibitem[Gurumurthy et~al.(2002)Gurumurthy, Bloem, , and Somenzi]{GBS02}
S.~Gurumurthy, R.~Bloem, , and F.~Somenzi.
\newblock Fair simulation minimization.
\newblock In \emph{CAV}, volume 2404 of \emph{LNCS}, pages 610--624. Springer,
  2002.

\bibitem[Henzinger et~al.(1995)Henzinger, Henzinger, and Kopke]{HHK:FOCS95}
M.~R. Henzinger, T.~A. Henzinger, and P.~W. Kopke.
\newblock {Computing simulations on finite and infinite graphs}.
\newblock In \emph{Foundations of Computer Science}, FOCS '95, Washington, DC,
  USA, 1995. IEEE Computer Society.
\newblock ISBN 0-8186-7183-1.
\newblock URL \url{http://portal.acm.org/citation.cfm?id=796255}.

\bibitem[Henzinger et~al.(2002)Henzinger, Kupferman, and
  Rajamani]{fairsimulation:02}
T.~A. Henzinger, O.~Kupferman, and S.~K. Rajamani.
\newblock {Fair Simulation}.
\newblock \emph{Information and Computation}, 173:\penalty0 64--81, 2002.
\newblock \doi{10.1006/inco.2001.3085}.
\newblock URL \url{http://dx.doi.org/10.1006/inco.2001.3085}.

\bibitem[Holzmann(2004)]{Holzmann:Spinbook}
G.~Holzmann.
\newblock \emph{The {SPIN} Model Checker}.
\newblock Addison-Wesley, 2004.

\bibitem[Jiang and Ravikumar(1991)]{ravikumar:hard:1991}
T.~Jiang and B.~Ravikumar.
\newblock {Minimal NFA Problems are Hard}.
\newblock In J.~Albert, B.~Monien, and M.~Artalejo, editors, \emph{ICALP},
  volume 510 of \emph{LNCS}, pages 629--640. 1991.
\newblock \doi{10.1007/3-540-54233-7\_169}.

\bibitem[Juvekar and Piterman(2006)]{piterman:generalized06}
S.~Juvekar and N.~Piterman.
\newblock {Minimizing Generalized B\"uchi Automata}.
\newblock In \emph{Computer Aided Verification}, volume 4414 of \emph{LNCS},
  pages 45--58. Springer-Verlag, 2006.
\newblock \doi{10.1007/11817963\_7}.
\newblock URL \url{http://dx.doi.org/10.1007/11817963\_7}.

\bibitem[Kupferman and Vardi(1996)]{kupfermanvardi:fair_verification}
O.~Kupferman and M.~Vardi.
\newblock {Verification of Fair Transition Systems}.
\newblock In \emph{Computer Aided Verification}, volume 1102 of \emph{LNCS},
  pages 372--382. Springer-Verlag, 1996.
\newblock URL \url{{\scriptsize
  http://citeseer.ist.psu.edu/viewdoc/summary?doi=10.1.1.29.9654}}.

\bibitem[Lee et~al.(2001)Lee, Jones, and Ben-Amram]{Lee:SCT2001}
C.~S. Lee, N.~D. Jones, and A.~M. Ben-Amram.
\newblock The size-change principle for program termination.
\newblock POPL '01, pages 81--92, 2001.
\newblock \doi{http://doi.acm.org/10.1145/360204.360210}.

\bibitem[Leroux and Point(2009)]{Talence:Presburger}
J.~Leroux and G.~Point.
\newblock {TaPAS: The Talence Presburger Arithmetic Suite}.
\newblock In \emph{Proceedings of the 15th International Conference on Tools
  and Algorithms for the Construction and Analysis of Systems (TACAS)}, volume
  5505 of \emph{LNCS}. Springer, 2009.

\bibitem[Niven(1965)]{Niven:1965}
I.~Niven.
\newblock \emph{Mathematics of Choice}.
\newblock The Mathematical Association of America, 1965.

\bibitem[Piterman(2006)]{Pit06}
N.~Piterman.
\newblock From nondeterministic {B\"uchi} and {Streett} automata to
  deterministic parity automata.
\newblock In \emph{LICS}, pages 255--264. IEEE, 2006.

\bibitem[Sebastiani and Tonetta(2003)]{Sebastiani-Tonetta:2003}
R.~Sebastiani and S.~Tonetta.
\newblock More deterministic vs. smaller {B\"uchi} automata for efficient {LTL}
  model checking.
\newblock In \emph{Correct Hardware Design and Verification Methods}, volume
  2860 of \emph{LNCS}, 2003.

\bibitem[Sistla et~al.(1987)Sistla, Vardi, and
  Wolper]{sistla:vardi:wolper:complementation:87}
A.~P. Sistla, M.~Y. Vardi, and P.~Wolper.
\newblock {The complementation problem for {B\"uchi} automata with applications
  to temporal logic}.
\newblock \emph{Theor. Comput. Sci.}, 49:\penalty0 217--237, Jan. 1987.
\newblock ISSN 0304-3975.
\newblock \doi{10.1016/0304-3975(87)90008-9}.
\newblock URL \url{http://dx.doi.org/10.1016/0304-3975(87)90008-9}.

\bibitem[Somenzi and Bloem(2000)]{somenzi:efficient}
F.~Somenzi and R.~Bloem.
\newblock {Efficient {B\"uchi} Automata from LTL Formulae}.
\newblock In \emph{Computer Aided Verification}, volume 1855 of \emph{LNCS},
  pages 248--263. Springer-Verlag, 2000.
\newblock \doi{10.1007/10722167\_21}.
\newblock URL \url{http://dx.doi.org/10.1007/10722167\_21}.

\bibitem[Tabakov and Vardi(2007)]{tabakov:model}
D.~Tabakov and M.~Vardi.
\newblock {Model Checking {B\"uchi} Specifications}.
\newblock In \emph{LATA}, volume Report 35/07. Research Group on Mathematical
  Linguistics, Universitat Rovira i Virgili, Tarragona, 2007.

\bibitem[Tsay et~al.(2008)Tsay, Chen, Tsai, Chan, and Luo]{GOAL_survey_paper}
Y.-K. Tsay, Y.-F. Chen, M.-H. Tsai, W.-C. Chan, and C.-J. Luo.
\newblock {GOAL} extended: Towards a research tool for omega automata and
  temporal logic.
\newblock In C.~Ramakrishnan and J.~Rehof, editors, \emph{Tools and Algorithms
  for the Construction and Analysis of Systems}, volume 4963 of \emph{LNCS},
  pages 346--350. 2008.
\newblock ISBN 978-3-540-78799-0.
\newblock URL \url{{http://dx.doi.org/10.1007/978-3-540-78800-3\_26}}.

\bibitem[Tsay et~al.(2011)Tsay, Tsai, Chang, and Chang]{buechistore:2011}
Y.-K. Tsay, M.-H. Tsai, J.-S. Chang, and Y.-W. Chang.
\newblock {B\"uchi} store: An open repository of {B\"uchi} automata.
\newblock In P.~Abdulla and K.~Leino, editors, \emph{Tools and Algorithms for
  the Construction and Analysis of Systems}, volume 6605 of \emph{LNCS}, pages
  262--266. 2011.
\newblock ISBN 978-3-642-19834-2.
\newblock URL \url{http://dx.doi.org/10.1007/978-3-642-19835-9\_23}.
\newblock 10.1007/978-3-642-19835-9\_23.

\end{thebibliography}
